\documentclass[runningheads]{llncs}
\usepackage[T1]{fontenc}

\usepackage{comment}
\usepackage{bm}
\usepackage{bbm}
\usepackage{tabularx}
\usepackage{float}
\usepackage{soul}
\usepackage{url}
\usepackage[utf8]{inputenc}
\usepackage[small]{caption}
\usepackage{graphicx}
\usepackage{amsmath}
\usepackage{xcolor}
\usepackage{booktabs}
\usepackage{algorithm}
\usepackage{algorithmic}
\usepackage{amssymb}
\usepackage{ltl}
\usepackage{wrapfig}
\usepackage{caption}
\usepackage[inline]{enumitem}
\usepackage{tablefootnote}
\usepackage{afterpage}
\usepackage{appendix}
\usepackage[hidelinks]{hyperref}
\usepackage{cleveref}
\hypersetup{
    colorlinks=true,
    linkcolor=blue,
    filecolor=black,      
    urlcolor=blue,
}
\usepackage{orcidlink}
\allowdisplaybreaks
\newcommand{\xmark}{$\times$}

\newtheorem{observation}{Observation}

\title{Sufficient Reasons and Explanations \\ for Reactive Systems\thanks{This work was supported by the Israel Science Foundation (ISF grant No. 655/25).}}

\author{Hadar Frenkel\orcidlink{0000-0002-3566-0338} \and
Nadav Rutman Moshe}

\authorrunning{H. Frenkel and N. Rutman Moshe}

\institute{Department of Computer Science and Artificial Intelligence, \\Bar-Ilan University, Israel\\
\email{hadar.frenkel@biu.ac.il}}

\begin{document}

%\linenumbers
\maketitle
\newcommand{\hadar}[1]{\begingroup\color{red}#1\endgroup}

\newcommand{\nadav}[1]{\begingroup\color{blue}#1\endgroup}
\newcommand{\quotes}[1]{``#1''}

\newcommand{\qonce}{\mathtt{p'\text{-}once}}
\newcommand{\ponce}{\mathtt{p\text{-}once}}
\newcommand{\qafter}{\mathtt{p'\text{-}after}}
\newcommand{\pafter}{\mathtt{p\text{-}after}}
\newcommand{\on}{\mathtt{on}}
\newcommand{\fresh}{\mathtt{f}}
\newcommand{\R}{\mathtt{r}}
\newcommand{\ap}{\textsc{AP}}
\newcommand{\ip}{\textsc{IP}}
\newcommand{\op}{\textsc{OP}}
\newcommand{\aug}{\mathtt{aug}}

\newcommand{\once}{\mathtt{once}}
\newcommand{\aut}{\mathcal{A}}
\newcommand{\cont}{\tau_{\text{ce}}}
\newcommand{\alltraces}{(2^\ap)^\omega}
\newcommand{\nat}{\mathbb{N}}

\newcommand{\phiTpi}{\tuple{\varphi, T, \pi}}
\newcommand{\SR}{SR}
\newcommand{\MSR}{MSR}
\newcommand{\MTSR}{MTSR}
\newcommand{\MSRs}{MSRs}
\newcommand{\MTSRs}{MTSRs}
\newcommand{\SRs}{SRs}
\newcommand{\TSR}{TSR}
\newcommand{\TSRs}{TSRs}
\newcommand{\TR}{\text{TR}}
\newcommand{\CE}{CE}
\newcommand{\CEs}{CEs}
\newcommand{\TCE}{TCE}
\newcommand{\TCEs}{TCEs}
\newcommand{\MTCE}{MTCE}
\newcommand{\MTCEs}{MTCEs}
\newcommand{\MCE}{MCE}
\newcommand{\MCEs}{MCEs}

\newcommand{\tuple}[1]{\ensuremath{\left \langle #1 \right \rangle }}
\newcommand{\set}[1]{\ensuremath{\{ #1 \} }}

\newcommand{\true}{\mathit{true}}
\newcommand{\false}{\mathit{false}}

\begin{abstract}
We address the problem of temporal causality and explainability for reactive systems, and, in this setting, study sufficient reasons and contrastive explanations. These two notions are well-known explainability measures in the context of neural networks. In this work, we unify these notions for reactive systems and formal specifications given in temporal logic, providing dedicated definitions for sufficient reasons and contrastive explanations. We then lift these definitions to \emph{temporal} sufficient reasons and contrastive explanations, providing more general and symbolic representations of explainability. We analyze the complexity of both verifying and finding explanations of the different types, and we demonstrate our approach using a prototype implementation.
\end{abstract}

\section{Introduction}

In this work, we address different aspects of explainability and causality studied in the literature, and study them in the context of reactive systems and temporal logic specifications. 
Reactive systems interact with their environment indefinitely, where, at each discrete time point, they can react to environmental inputs with some system outputs. They form the common formalism used in the verification community today, used as an abstraction for many different formalisms, from  hardware circuits~\cite{BAKHTERI20201999}, to software code~\cite{10.1007/978-3-642-12029-9_18}. 
The prominent logic to reason about the behavior of reactive systems over time is Linear Temporal Logic (LTL, \cite{Pnueli77}), widely used to formally prove that systems adhere to their specifications. But,
even if we can prove the correctness or incorrectness of a system, 
it is not always clear \emph{why} a system satisfies or falsifies its specification. 

In recent years, many works have studied explainability and causality in reactive systems~\cite{Leitner-FischerL13a,BaierDFJMPZ21,CaltaisGL18}. In particular, a recent line of work has studied \emph{temporal actual causality},
explaining effects given as LTL specifications: Based on the definitions of Halpern and Pearl \cite{HalpernP05a}, 
\cite{coenen2022hyperproperty} provides an algorithm for finding actual causes as finite sets of events on the trace;  
\cite{DBLP:conf/atva/CoenenFFHMS22} defines causes as temporal properties; \cite{FFMS24} presents an algorithm for the synthesis of temporal actual causes as automata over infinite words; and \cite{LICS25} provides a complexity analysis. 

In this work, we study different formalisms for temporal explainability, namely, \emph{sufficient reasons} and \emph{contrastive explanations}~\cite{DBLP:journals/pacmmod/BarceloKRSV25}, which are intuitive formalisms that are already used to explain the outputs of neural networks. 
Similar to Halpern and Pearl's actual causality, sufficient reasons take into account the \emph{current} input to the system, and aim at explaining the behavior of the system, when given this input. 
But,
while Halpern and Pearls' actual causes are the most refined reasons for a given effect, sufficient reasons provide the most general reason for the effect (though they might collide in some cases). For example,~\cite{HalpernP05a} describe a scenario in which Suzy and Billy throw a rock, but since Suzy is stronger, her rock hits a bottle first, which causes the bottle to break. In their actual causality analysis, they define an elaborated mechanism to derive that Suzy is the actual cause for the broken bottle. 
\emph{Sufficient reasons} define the sufficient events,
that, had they occurred, the result would stay the same.
%which, given they had happened, the result would stay the same. 
In this case, \emph{Suzy and Billy throw the rock} is a sufficient reason, but minimal sufficient reasons would be \emph{Suzy throws the rock}, or \emph{Billy throws the rock}. While actual causality can be used to derive accountability (here, Suzy is to blame for the broken bottle), (minimal) sufficient reasons tell us what are the sufficient actions we need to take in order to achieve the same outcome. 
\emph{Contrastive explanations}, that we study at the end of this paper, describe what are the minimal changes we need to do to avoid the outcome. 

In the context of reactive systems, when a system violates its specification, a model-checker outputs a counterexample trace \cite{DBLP:conf/cav/BraytonM10}. Applying our causal analysis on this trace using sufficient reasons, we can find the most general behavior of the system that would still violate the specification (analysis which one can then use to repair the system). In contrast, we can use contrastive explanations to find minimal changes that would help avoiding the violation. 
In the same way, we can find positive example traces that satisfy the specification, and use 
our proposed analysis to derive what actions would maintain the same positive behavior (using sufficient reasons), or what are minimal changes that would lead to a violation, to help us avoid it (using contrastive explanations). 

Motivated by the need for such intuitive and simple definitions of explainability, 
%
%for (minimal) sufficient reasons, as well as for contrastive explanations, 
we provide definitions of (minimal) sufficient reasons and contrastive explanations for reactive systems, both as sets of events (not limited to finite sets), and as temporal reasons and explanations. We present an algorithmic approach for verifying and finding (events/temporal) reasons and explanations, and demonstrate our approach with a prototype implementation, and we provide complexity analysis for each problem: verification/synthesis of sufficient reasons/contrastive explanations in the temporal/non-temporal setting.

Following these definitions, this work presents a new approach for temporal explainability for reactive systems, allowing the user to choose what definitions to use depending on the type of desired explanation. Moreover, we establish a common ground for a formal definition of explainability, already used across different fields.

\paragraph{Contributions}
We provide definitions, complexity analysis, and algorithms for sufficient reasons and contrastive explanations for effects given in temporal logics. We show that for minimal sufficient reasons that are sets of events on the trace, there is always an ultimately periodic representation for this set of events, which we can find by solving a model-checking instance.
Our use of logical formalism to express minimal sufficient reasons allows us to find \emph{infinite} sets of events on a trace.
For minimal temporal sufficient reasons, that is, the reasons themselves are temporal properties and not only sets of events, we show that there is always a unique, $\omega$-regular minimal temporal sufficient reason. We show that the problems of verifying (minimal, temporal) sufficient reasons and contrastive explanations are all PSPACE complete, which stems from the underlying LTL model-checking problem. 
The main focus of our work is the analysis of sufficient reasons. In addition, in \Cref{sec:ContrastiveExplanations}, we briefly address contrastive explanations. While contrastive explanations were addressed in previous works~\cite{FS23,DBLP:journals/corr/abs-2310-01001} as counterfactual traces, here we provide a complexity analysis for the synthesis of contrastive explanations, and provide a characterization of \emph{temporal} contrastive explanations, that were not studied before. 
Finally, we present a prototype implementation. %that demonstrates our approach. 

\section{Preliminaries}\label{sec: prelim}
\subsubsection*{Transition Systems and Temporal Logic} 
Given a finite set $X$, $X^*$ is the set of all finite sequences of elements of $X$, $X^+$ is the set of all non-empty finite sequences, and $X^\omega$ is the set of all infinite sequences of elements of $X$.
\paragraph{Transition Systems and Traces}
Let ${\ap}$ be a finite set of atomic propositions, such that $\ap = \ip \cup \op$ is the union of input propositions ($\ip$) and output propositions $(\op)$ with $\ip \cap \op = \emptyset$. 
We model a reactive transition system as a Mealy Machine~\cite{Mealy} $T = (Q, q_0, \delta, \ip, \op)$ where 
%
%A \emph{labeled transition system} (LTS) $T = (Q, Q_0, \delta, L)$ is a tuple where 
$Q$ is a finite set of states, $q_0\in Q$ is the initial state, and 
$\delta: Q\times 2^\ip \rightarrow Q\times 2^\op$ is the transition function, mapping the current state and set of inputs to the next state and a set of outputs. 
%$\delta \subseteq Q\times Q$ is a transition relation, and $L: Q \rightarrow 2^{\ap}$ is a labeling function, that labels each state with a subset of atomic propositions.
$T$ is \emph{input deterministic}, that is, $\forall q\in Q, \forall i\in 2^\ip$, there exists exactly one pair $(q',o)$ such that $\delta(q,i) = (q',o)$.

A \emph{path} of $T$ is an infinite sequence of states $p_0 p_1 p_2 \cdots \in Q^\omega$ such that $p_0 = q_0$ and $\forall j: \exists i_j \in 2^\ip,o_j\in 2^\op$ such that $\delta(p_j, i_j) = (p_{j+1}, o_j)$.
The \emph{trace} induced by such a path is 
$\pi_0 \pi_1 \pi_2 \cdots \in (2^\ap)^\omega$ where $\forall j: \pi_j = i_j \cup o_j$.

 $\TR(T)$ is the set of all traces induced by paths of a transition system $T$.

We say that a trace $\pi \in (2^{\ap})^\omega$ is \emph{ultimately periodic} if there exist $u \in (2^{\ap})^*, u \in (2^{\ap})^+$ such that $\pi = u\cdot v^\omega$.

For an infinite trace $\pi$, $\pi[j]$ denotes the $j$'th location of $\pi$, and $\pi[j,\infty)$ denotes the suffix of $\pi$ starting from position $j$.

The \emph{projection} of a trace $\pi \in (2^\ap)^\omega$ on a subset $\ap' \subset \ap$ is $\pi_{\downarrow_{\ap'}} = (\pi[0] \cap \ap') (\pi[1] \cap\ap' )\cdots \in (2^{\ap'})^\omega $.

For traces $\pi_1 \in (2^{\ap_1})^\omega$, $\pi_2 \in (2^{\ap_2})^\omega$ we define their union $\pi_1 \cup \pi_2$ as dot-wise union, that is $\pi_1 \cup \pi_2 = ( \pi_1[0] \cup \pi_2[0] ) ( \pi_1[1] \cup \pi_2[1] )( \pi_1[2] \cup \pi_2[2] )\cdots $.

\paragraph{Büchi Automata}
A Non-deterministic Büchi Automaton (NBA) over infinite words is $ \aut = (\Sigma, Q, q_0, \delta, F)$ where $\Sigma$ is a finite alphabet, $Q$ is a finite set of states and $q_0 \in Q$ is the initial state, $\delta \subseteq Q\times\Sigma\times Q$ is the transition relation, and $F\subseteq Q$ is the set of accepting states. 

A \emph{run} of an NBA $\aut$ on a word $\alpha_0\alpha_1\alpha_2 \cdots\in \Sigma^\omega$ is an infinite sequence $p_0 p_1 p_2 \cdots\in Q^\omega$ such that $p_0=q_0$ and $\forall j: (p_j, \alpha_j, p_{j+1}) \in \delta$. 

An NBA $\aut$ accepts a word $\alpha$ if there is a run $p_0 p_1 p_2\cdots$ of $\aut$ on $\alpha$ such that for infinitely many $j$'s we have $p_j\in F$. The \emph{language} of $\aut$, denoted by $\mathcal{L}(\aut)$, is the set of all words accepted by $\aut$.

We can view any Mealy machine 
$T = (Q, q_0, \delta, \ip, \op)$ 
as an NBA $ \aut = (2^\ap, Q, q_0, \delta', F)$ such that $F = Q$ and $\delta' = \{  (q,i \cup o, q') ~|~ i\in 2^\ip, o\in 2^\op, \delta(q, i) = (q', o)  \}$. Therefore, in the rest of the paper, we perform operations such as intersection between an NBA and a transition system, assuming an automaton representation for the system.

\paragraph{LTL and HyperLTL} 
We define Linear Temporal Logic (LTL)~\cite{Pnueli77} and its relational extension HyperLTL~\cite{ClarksonFKMRS14}. 
The goal of this paper is to provide reasons and explanations for the behavior of a system with respect to LTL specifications. We use HyperLTL as a meta-logic to express reasons and explanations. Both logics are defined over the set of atomic propositions $\ap$. HyperLTL is also defined over a finite set $\mathcal{V}$ of trace variables. 
The definition of LTL is given in black using the following grammar. The \textcolor{blue}{blue} parts are the HyperLTL extension. 
%\begin{equation*}
\begin{align*}
 & \psi := a ~|~ \psi \vee \psi ~|~ \neg\psi ~|~ \LTLnext \psi ~|~ \psi\LTLuntil \psi 
  & \textcolor{blue}{\varphi := \exists \pi. \varphi ~|~ \forall \pi. \varphi ~|~ \psi} 
\end{align*}
%\tag{$\star$}\label{eq:hyperLTL}
%\end{equation*}
%
Where $a\in \ap$ for LTL formulas, and $a\in \{ p_\pi ~|~ p\in\ap, \pi\in \mathcal{V} \}$ for HyperLTL formulas. That is, HyperLTL formulas are LTL formulas where atomic propositions are augmented by trace variables, and with a quantifier prefix. 
%, with a prefix of existential and universal trace quantifiers. 

In addition,
$\LTLnext$ is the temporal operator \emph{next} and $\LTLuntil$ is the temporal operator \emph{until}. We use the usual derived Boolean operators $\wedge, \rightarrow, \leftrightarrow$, and the derived temporal operators $\LTLeventually \psi = \mathit{true}\LTLuntil \psi$  (\emph{eventually} $\psi$) and $\LTLglobally \psi= \neg \LTLeventually \neg \psi$ (\emph{globally}~$\psi$). 

\emph{LTL Semantics:}
We define when an infinite word $\alpha\in (2^\ap)^\omega $
satisfies an LTL formula $\varphi$ inductively, as follows: 
\begin{align*}
       &\alpha \vDash p&\text{ iff } & p\in \alpha[0] \\ 
    &\alpha \vDash \neg \psi &\text{ iff } &\alpha \nvDash \psi \\
    & \alpha\vDash \psi_1 \vee \psi_2  &\text{ iff } &  \alpha \vDash \psi_1 \text{ or } \alpha \vDash \psi_2 \\
       &\alpha \vDash \LTLnext \psi &\text{ iff } & \alpha[1,\infty) \vDash \psi \\ 
      &\alpha \vDash \psi_1 \LTLuntil \psi_2 &\text{ iff } & \exists t_2\in\mathbb{N} \text{ s.t. } \alpha[t_2,\infty) \vDash \psi_2 \text{ and } \forall t_1< t_2: \alpha[t_1,\infty) \vDash \psi_1 
\end{align*}
%We say that a 
A system $T$ satisfies an LTL formula $\varphi$, denoted by $T\vDash \varphi$ iff $\forall \alpha \in \TR(T): \alpha\vDash\varphi$.

The language of an LTL formula $\varphi$ is the set of all words that satisfy $\varphi$: 
$\mathcal{L}(\varphi) = \{ \alpha\in (2^\ap)^\omega~|~ \alpha\vDash\varphi \} $. 
Every LTL formula $\varphi$ can be translated into an NBA $\aut$ such that $\mathcal{L}(\varphi) = \mathcal{L}(\aut)$~\cite{DBLP:conf/lics/VardiW86}. We denote this NBA by $\aut_\varphi$ and will use this notation throughout the paper. 

\emph{HyperLTL Semantics} are given with respect to a set of traces $\mathcal{T}$ and a trace assignment $\Pi: \mathcal{V}\rightarrow \mathcal{T}$. Given a trace variable $\pi\in\mathcal{V}$ and a trace $\tau\in\mathcal{T}$, we denote by $\Pi[\pi\mapsto \tau]$ the trace assignment that coincides with $\Pi$ on all variables, except for $\pi$ which is mapped to $\tau$. 
For the assignment $\Pi$, we denote by $\Pi[t, \infty)$ the assignment that maps every variable $\pi \in \mathcal{V}$ to the suffix of $\Pi(\pi)$ starting from position $t$, that is: $\Pi(\pi)[t]\Pi(\pi)[t+1]\Pi(\pi)[t+2]\cdots$. 
We define the semantics of HyperLTL only for atomic propositions and for the quantifiers, since the semantics of the Boolean and temporal operators are defined in the same way as for LTL. We refer the reader to \cite{ClarksonFKMRS14}, or to Appendix~\ref{app:HyperLTL}, for the full definition.
\begin{align*}
     &\Pi, \mathcal{T} \vDash p_\pi &\text{ iff } & p\in\Pi(\pi)(0) \\ 
   % &\Pi, \mathcal{T} \vDash \neg \psi &\text{ iff } &\Pi,\mathcal{T} \nvDash \psi \\
   % & \Pi, \mathcal{T} \vDash \psi_1 \vee \psi_2  &\text{ iff } &  \Pi, \mathcal{T} \vDash \psi_1 \text{ or } \Pi, \mathcal{T} \vDash \psi_2 \\
    %   &\Pi, \mathcal{T} \vDash \LTLnext \psi &\text{ iff } & \Pi(\pi)[1,\infty),\mathcal{T} \vDash \psi \\ 
    %  &\Pi, \mathcal{T} \vDash \psi_1 \LTLuntil \psi_2 &\text{ iff } & \exists t_2\in\mathbb{N} \text{ s.t. } \Pi[t_2,\infty),\mathcal{T} \vDash \psi_2 \\ &&&\text{ and } \forall t_1< t_2: \Pi[t_1,\infty),\mathcal{T} \vDash \psi_1 \\ 
      &\Pi, \mathcal{T} \vDash \exists \pi. \varphi &\text{ iff } & \text{ there exists } \tau\in\mathcal{T} \text{ s.t. } \Pi[\pi\mapsto \tau],\mathcal{T}\vDash \varphi \\
       &\Pi, \mathcal{T} \vDash \forall \pi. \varphi &\text{ iff } & \text{ forall } \tau\in\mathcal{T}: \Pi[\pi\mapsto \tau],\mathcal{T}\vDash \varphi 
\end{align*}
We say that $\mathcal{T} $ satisfies $\varphi$, denoted $\mathcal{T} \vDash \varphi$, if $\Pi_\emptyset, \mathcal{T} \vDash \varphi$ where $\Pi_\emptyset: \mathcal{V}\rightarrow \emptyset$~is~the empty assignment, meaning that $\varphi$ contains no free trace variables that need to be assigned.
A system $T$ satisfies a HyperLTL formula~$\varphi$ if $\TR(T) \vDash\varphi$.

\paragraph{Trace Properties}
A trace property $P \subseteq (2^\ap)^\omega$ is a set of traces. For example, the language of an LTL formula, and the language of an NBA, are trace properties. 
We say that a word $\alpha$ satisfies a trace property $P$ if $\alpha \in P$. 

If $P$ is over a subset of $\ap' \subset \ap$, i.e., $P \subseteq (2^{\ap'})^\omega$, we say that a word $\alpha \in (2^\ap)^\omega$ satisfies $P$ if the projection of $\alpha$ over $\ap'$ satisfies $P$, i.e.,
$\alpha_{\downarrow_{\ap'}} \in P$.

\paragraph{Ultimately Periodic Sets}
A set $R \subseteq \mathbb{N}\times \ap$ is \emph{ultimately periodic} if there exist $n_0,p \in \mathbb{N}$ such that forall $n>n_0$, we have that $\tuple{n,a} \in R $ iff $\tuple{n+p,a} \in R$ forall $a\in \ap$
(in particular, if $R$ is finite then it is ultimately periodic). This way, ultimately periodic sets, even if infinite, yield a finite representation.

\paragraph{Additional Notations}
\begin{itemize*}[nosep]
        \item An \emph{event} is a pair $\tuple{n,a}\in\nat\times \ap$.\\
        \item For an LTL formula $\varphi$ and a trace $\pi$, we denote by $\varphi_\pi$ the formula $\varphi$ where all atomic propositions are parameterized by $\pi$. E.g., $(\LTLglobally p \vee \LTLeventually q)_\pi := \LTLglobally p_\pi \vee \LTLeventually q_\pi$.\\
    \item We denote a HyperLTL %or HyperQPTL 
    formula $\varphi$ with free variables $\pi_1, \ldots, \pi_n$ by $\varphi[\pi_1, \ldots, \pi_n]$.\\
    \item For traces $\pi$ and $\pi'$, and a set $A\subseteq \ap$ we use 
      $\pi =_{A} \pi'$ to denote $\bigwedge_{p\in A} p_\pi \leftrightarrow p_{\pi'}$, i.e., the two traces agree on all propositions in $A$ at a certain timestep. If we want to explicitly refer to a timestep $t$, we write $\pi[t] =_{A} \pi'[t]$. 
    %to denote that the two traces agree on all propositions in $A$  at timestep $i$, that is: $\bigwedge_{p\in A} p\in \pi[i] \leftrightarrow p\in \pi'[i]$. Note that this is not a HyperQPTL formula. 
    If $A = \ap$ we simply write $\pi = \pi'$. We use ${\pi \equiv_{A} \pi'}$ to denote that the traces are globally equal, that is: ${\pi \equiv_{A} \pi'} := \LTLglobally (\pi =_{A} \pi')$. %This generalizes the above formula for all timesteps. 
    Again, if $A=\ap$ we write $\pi \equiv \pi'$.% \hadar{remove this? we need to many equality notations}
\end{itemize*}
%\hadar{need to add that the reason actually holds, otherwise everything can be a reason vacuously}
\subsubsection*{Explanations}
%In the following, 
We present the definitions as given in~\cite{DBLP:journals/pacmmod/BarceloKRSV25}, with slight modifications for ease of readability. 
\paragraph{Sufficient Reasons (Abductive Explanations)}\label{sec:prelin:suff}
Given a function $f: X^n \rightarrow Z$, a sufficient reason for $f(\overline{x}) = z$ is a subset $R\subseteq [n] $ such that for every $\overline{y} \in X^n$: if $\forall j\in R: y[j] = x[j]$  then $f(\overline{x}) = f(\overline{y})$ ($\overline{x}[j]$ denotes the $j$'th component of $\overline{x}$).

$R$ is a \emph{minimal sufficient reason} if it is a smallest set in terms of containment, i.e., there is no sufficient reason $R'$ with $R' \subsetneq R$. 

\paragraph{Contrastive Explanations (Counterfactual Explanations)}
A contrastive explanation for $f(\overline{x}) = z$ is a vector $\overline{y}\in X^n$ such that $f(\overline{x}) \neq f(\overline{y})$.\footnote{In particular, following the definition of \cite{DBLP:journals/pacmmod/BarceloKRSV25}, sufficient reasons are of a different type than contrastive explanations, as the former is given as a set of components of the domain, while the latter is a specific vector.} 

%\emph{Minimal Contrastive Explanations}
%We would usually look for a \emph{close} contrastive explanation, therefore, 
Given a distance measure $d :X^n \times X^n \rightarrow \mathbb{N}$, a \emph{minimal contrastive explanation} $\overline{y}\in X^n$
is a contrastive explanation, such that for all contrastive explanations $\overline{y'}\in X^n$
we have $d(\overline{x},\overline{y}) \leq d(\overline{x},\overline{y'})$.

\section{Running Example}\label{sec: example}

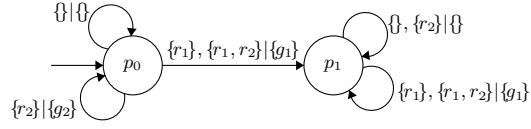
\begin{figure}[t]
    \centering
    
\begin{tikzpicture}[scale=0.13,  every node/.append style={scale=0.8}]
\tikzstyle{every node}+=[inner sep=0pt]
\draw [black] (24.2,-23.4) circle (3);
\draw (24.2,-23.4) node {$p_0$};
\draw [black] (44.6,-23.4) circle (3);
\draw (44.6,-23.4) node {$p_1$};
\draw [black] (15.9,-23.4) -- (21.2,-23.4);
\fill [black] (21.2,-23.4) -- (20.4,-22.9) -- (20.4,-23.9);
\draw [black] (21.722,-21.73) arc (263.75472:-24.24528:2.25);
\draw (18,-19.03) node [above] {$\{\!\}|\{\!\}$};
\fill [black] (24.02,-20.42) -- (24.6,-19.68) -- (23.61,-19.57);
\draw [black] (27.2,-23.4) -- (41.6,-23.4);
\fill [black] (41.6,-23.4) -- (40.8,-22.9) -- (40.8,-23.9);
\draw (34.4,-22.9) node [above] {$\{\!r_1\!\},\{\!r_1,r_2\!\}|\{\!g_1\!\}$};
\draw [black] (23.341,-26.262) arc (11.03438:-276.96562:2.25);
\draw (15.18,-26.82) node [below] {$\{\!r_2\!\}|\{\!g_2\!\}$};
\fill [black] (21.41,-24.46) -- (20.52,-24.12) -- (20.72,-25.1);
\draw [black] (45.61,-20.587) arc (187.98519:-100.01481:2.25);
\draw (49.98,-19.24) node [right] {$\{\!\},\{\!r_2\!\}|\{\!\}$};
\fill [black] (47.45,-22.49) -- (48.31,-22.87) -- (48.17,-21.88);
\draw [black] (47.577,-23.657) arc (112.79507:-175.20493:2.25);
\draw (51.01,-26.17) node [right] {$\{\!r_1\!\},\{\!r_1,r_2\!\}|\{\!g_1\!\}$};
\fill [black] (46.21,-25.92) -- (46.06,-26.85) -- (46.98,-26.46);
\end{tikzpicture}

    \caption{A transition system over $\ip = \{\!r_1, r_2\!\}, \op = \{\!g_1, g_2\!\}$}
    \label{fig:motivation}
\end{figure}

%Consider the system of \Cref{fig:motivation}. This is a simple system, but it describes an interesting (and undesired) behviour, that will help us demostrate the motivation for sufficient reasons and contrastive explanations. 

Consider the system of \Cref{fig:motivation}, representing a biased arbiter: once $\mathsf{user}1$ issues a request ($r_1$), $\mathsf{user}2$ can never be granted access to some resource ($g_2$). On the other hand, every request of $\mathsf{user}1$ is answered by a grant ($g_1$), and requests of $\mathsf{user}2$ are answered by grants as long as $\mathsf{user}1$ did not issue any request. 

\textbf{(Minimal) sufficient reasons:}
Consider the simple specification \emph{$\mathsf{user}2$ is eventually granted access}, expressed by the LTL formula $\varphi_1 = \LTLeventually g_2$, and the counterexample trace $\pi = \{r_1, r_2, g_1\}^\omega$. A minimal sufficient reason for the violation of $\varphi_1$ on $\pi$ would be a minimal set of events that appear on $\pi$ and are sufficient to determine that $\pi \vDash \neg \varphi_1$. In this case,
$\{\langle 0,r_1\rangle\}$, 
is such a minimal sufficient reason: 
 %every trace on which $\mathsf{user}1$ issues a request on the first timestep, would never have $g_2$ -- 
it is sufficient to observe $r_1$ on time $0$, to never observe $g_2$. 
A sufficient reason for $\neg \varphi_1 := \LTLglobally \neg g_2$ would be any superset of $\{\langle 0,r_1\rangle\}$. 

\textbf{Temporal sufficient reasons:}
%Now, consider the specification \emph{every request of $\mathsf{user}2$ is eventually answered by a grant}, expressed using the LTL formula $\varphi_2 = \LTLglobally(r_2 \rightarrow \LTLeventually g_2)$, with the same counterexample trace $\pi$. We look for a sufficient reason for the violation of $\varphi_2$, that is, $\LTLeventually (r_2 \wedge \LTLglobally \neg g_2)$
Now, assume we want $\mathsf{user}1$ to never be granted access, via the LTL specification $\LTLglobally \neg g_1$. It's negation is $\LTLeventually g_1$, and the trace $\pi$ above is a witness for the violation. 
%consider the same trace $\pi$ with the effect   $\varphi_{{1}}\ = \LTLeventually r_1$. 
Every occurrence of $r_1$ on $\pi$ constitutes a sufficient reason for $\LTLeventually g_1$. However, in this case, we might want to describe the reason in a more general way, as a temporal property. A temporal sufficient reason would be any property that $\pi$ satisfies which is sufficient to observe $\LTLeventually g_1$. 
Here, temporal sufficient reasons are: $\LTLglobally r_1$ -- every trace that always has $r_1$ (as $\pi$), would satisfy $\LTLeventually g_1$; $\LTLglobally \LTLeventually r_1$; simply $r_1$; and $\LTLeventually r_1$, which is also the minimal temporal sufficient reason (see \Cref{sec:TemporalReasons}): While the trace~$\pi$ satisfies $\LTLglobally r_1$, it is sufficient that it would have only satisfied $\LTLeventually r_1$ to observe~$\LTLeventually g_1$.

\textbf{Contrastive explanations:} For $\LTLeventually g_1$ and $\pi$ above, a contrastive explanation is a closest trace $\tau$ to $\pi$, on which $\LTLeventually g_1$ does not hold,
%, but intuitively, they are dual to sufficient reasons. 
for example
%for a contrastive explanation is 
$\{ r_2, g_2\}^\omega$: only avoiding all occurrences of $r_1$ would avoid $\LTLeventually g_1$. %\footnote{Recall that contrastive explanations are specific traces and not sets of events.}

\section{Definitions}
%\begin{enumerate}[nosep]
  %  \item (a) Find a set of events in $\pi$ which serve as a (minimal) sufficient reason for the fact that $\pi\models\varphi$; or (b) Given such set of events, verify that it is a (minimal) sufficient reason.  
   % \item Find or verify a symbolic representation as a trace property, for a (minimal) sufficient reason for the fact that $\pi$ satisfies $\varphi$. 
    %\item Similarly, find or verify (symbolic) contrastive explanations for the fact that $\pi$ satisfies $\varphi$. 
%\end{enumerate}
%In the following, 
We define specific traces and structures that we use throughout the paper. The motivations for these constructs appear when we use them, later on in the paper. %(one can also skip this, and come back to it when the corresponding definition is needed).  
\begin{definition}[Trace representation]
\label{def:tracerepresentation}
    Let $R\subset \mathbb{N}\times \ip$ and let $\R \notin \ap$.
    We represent $R$ as a trace $\tau_R$ over the alphabet $(2^{\{\R\}\cup\ip})^\omega$ as follows:
    $\forall t\in\nat$: $\tau_R[t] = \{\R \} \cup \{ a~|~\tuple{t,a} \in R   \}$.\footnote{We use the atomic proposition $\R$ to indicate that $\tau_R$ is used as a trace representation, and not as a system trace. We elaborate on this in \Cref{def:augT}. }  That is, $R = \{\tuple{t,a}~|~ a\in \tau_R[t] \}$, and $\forall t\in\nat: \R\in \tau_R[t]$. We call $\tau_R$ the \emph{trace representation} of $R$. 
If $R$ is ultimately periodic then $\tau_R$ is an ultimately periodic trace.
    %
    %From now on, we assume that the set of inputs is $\ip \cup \{\R, \on \}$.
\end{definition}
\begin{example}
    Let $\ip = \{ i_1, i_2\}$. The trace representation of
    $R_1 = \{ \tuple{0, i_1}, \tuple{ 1, i_2}\}$ is $\tau_{R_1} = \{ \R, i_1\}\set{\R, i_2}\set{\R}^\omega $. For $R_2 = \{ \tuple{t, i_1}| t\%2\! =\!0 \}$ we have $\tau_{R_2} = (\set{\R,i_1}\set{\R})^\omega$.
\end{example}

\begin{definition}[Augmented transition system]\label{def:augT}
        %old definition in the main_old file 
     Let $T = (Q, q_0, \delta, \ip, \op)$ be a transition system over $\ap = \ip \cup \op$. We define the \emph{augmented transition system} $T_\aug$ over inputs $\ip' = \ip \cup \{ \R\}$ and outputs $\op$ %of atomic propositions, 
     to be  
     $ T_\aug = (Q\cup \{q_\R \}, q_0, \delta', \ip', \op )$
    % \begin{align*}
     %   T_\aug = (Q\cup \{q_\R \}, q_0, \delta', \ip', \op )
    %\end{align*}
        for $q_\R\notin Q$, where:
        \begin{itemize}[nosep]
            \item    $\forall q\in Q, A\in 2^\ip: \delta'(q, A) = \delta(q,A)$.
            %\item  $L'(q_{ \on}) = \{\R, \on\}, L'(q_{ \neg\on}) = \{\R\}$.
            %\item $\delta'  = \delta \cup \{ (q_\on, q_\on), (q_{\neg\on}, q_{\neg\on}), (q_\on, q_{\neg\on}), (q_{\neg\on}, q_\on) \}$
            \item  $\forall A\in 2^\ip$:
            $\delta' (q_0, \{\R\}\cup A)  = (q_\R, \emptyset)= \delta' (q_\R, \{\R\}\cup A)$.
            %  \item $ \forall A\in 2^\ip$: $\delta' (q_\R, \{\R\}\cup A) = (q_\R, \emptyset)$.
        \end{itemize}
   
\end{definition}
  That is, 
  given a system $T$, $T_\aug$
  %the augmented transition system of a system $T$ 
  contains all traces of $T$, and, in addition, it contains all traces $\tau \in (2^{\{\R\}\cup \ip})^\omega$ that can serve as a trace representation of some set $R$. 
Although we add exponentially many transitions in the size of $\ip$, the size of $T_\aug$ is linear in the size of $T$, since the alphabet of~$T$ is $2^{\ip\cup \op}$.

\begin{definition}[Projection and Extension of Büchi automata]\label{def:projection}\label{def:extension}
Let $ \aut = (2^\ap, Q, q_0, \delta, F)$ be an NBA. 

The \emph{projection} of $ \aut$ on a subset $\ap' \subseteq \ap$ is $\aut_{\downarrow_{\ap'}} = (2^{\ap'}, Q, q_0, \delta', F)$ where $\delta' = \{  (q, A, q') ~|~ A\in 2^{\ap'} \text{ and }
    \exists B\in 2^{\ap \setminus \ap'}: (q, A\cup B, q')\in \delta  \}$.

    The \emph{extension}  of $ \aut$ to a superset $\ap'' \supseteq \ap$ is $\aut_{\uparrow_{\ap''}} = (2^{\ap''}, Q, q_0, \delta'', F)$ where $\delta'' = \{  (q, A\cup B, q') ~|~ A\in 2^{\ap},B \in 2^{\ap'' \setminus \ap}, 
     (q, A, q')\in \delta  \}$.
\end{definition}

\subsubsection*{Problem Definition}
Given a transition system $T$, an LTL specification $\varphi$, which we call the \textbf{effect}, and an ultimately periodic trace $\pi\in\TR(T)$ with $\pi \models \varphi$,\footnote{In \Cref{sec: example}, we looked for a reason for a violation of an LTL property. For simplicity, from now on, we assume that $\pi\vDash \varphi$. If we are looking for a reason for $\pi\nvDash \varphi$, we can simply look for a reason for the formula $\varphi' = \neg \varphi$, which $\pi$ satisfies.} we are interested in the following:
\begin{itemize}[nosep]
    \item \textbf{Reason type}: We study reasons first as sets of events on the trace, and then via symbolic representation as trace properties.
    \item \textbf{Verification \& synthesis:} For each reason type, we are interested in the \emph{verification} problem: whether a suspected reason is indeed a reason; and in the \emph{synthesis} problem: finding a (minimal) reason. 
\end{itemize} 
  
\section{Non-Temporal Sufficient Reasons}\label{sec:nontemporal}
%As described in Section~\ref{sec:prelin:suff}, 
A sufficient reason provides a set of elements that is sufficient to fix the output of a function. In this section, we 
lift this definition to express sufficient reasons for LTL effects in reactive systems. 
%\subsection{Sufficient Reasons for Reactive Systems}
%\paragraph{Sufficient Reasons over Traces}
We define reasons as (possibly infinite) sets of events that appear on the trace. In \Cref{sec:TemporalReasons} we lift the definition once more, to refer to \emph{temporal} sufficient reasons. 
In the following, we define reasons over the set $\ip$ of input propositions, while the effect $\varphi$ can be over all of $\ap$, to express relations between inputs and outputs. One can easily tweak the definition to take into account other propositions either in the reason or in the effect.

\begin{definition}[Sufficient Reasons for LTL Effects]\label{def:SR}
Let $T$ be a transition system,
 $\pi \in \TR(T)$, $\varphi$ an LTL formula such that $\pi\vDash\varphi$, and $R\subseteq \mathbb{N}\times \ip$.
   \begin{itemize}[nosep]
       \item We say that $R$ is a \emph{sufficient reason} (SR) for 
       $\pi \vDash \varphi$ 
       %$f_\varphi(\pi) = b$ (where $b\in \{0,1\}$) 
       in $T$, if for every trace $\tau \in \TR(T)$, if $\tau[t] =_{\{a\}} \pi[t]$ for all $\tuple{t,a}\in R$, then $\tau \vDash \varphi$. %\footnote{In general, $R$ can be empty. However, this would require that \emph{all} traces satisfy the effect, which is not usually the case.} 
       In that case, we say that $R$ is a \SR~for $\phiTpi$. 
   % \item We say that $R$ is a \emph{minimum sufficient reason}  for $\phiTpi$     if $R$ is a sufficient reason for $\phiTpi$ and there is no sufficient reason $R'$ with $|R'| < |R|$. This covers also  the case where all sufficient reasons are countable. 
    \item We say that $R$ is a \emph{minimal sufficient reason} (MSR)  for $\phiTpi$ 
     if $R$ is a \SR~for $\phiTpi$ and 
    there is no \SR~$R'$ such that $R' \subsetneq R$. 
   \end{itemize} 
\end{definition}

 Since we consider input-deterministic Mealy machines, a trace $\tau$ that agrees with $\pi$ on all inputs at all timepoints is only $\pi$ itself, and thus satisfies $\varphi$. Therefore we have: 
\begin{observation}\label{thm:NNSR}\label{obs:NNSR}
    Let $T$ be a transition system, let $\varphi$ be an LTL formula, and let $\pi \in \TR(T)$ such that $\pi\vDash\varphi$. Then, $R_\text{all} = \mathbb{N}\times \ip$ is a \SR~for $\phiTpi$. Moreover, we can finitely represent $R_\text{all}$ using the trace representation $\tau_{R_\text{all}} = (\{ \R\} \cup \ip)^\omega$.
\end{observation}

Note that the definition of \MSRs~allows us to compare reasons of countable size. However, we might have an infinite chain $R_1\! \supsetneq \!R_2 \!\supsetneq \!R_3 \!\supsetneq\!\cdots$ with no \MSR.

\begin{example}\label{ex:infinite}
    Consider the transition system of \Cref{fig:motivation}, the
    trace $\pi = \{r_1, g_1\}^\omega$ and the effect $\varphi = \LTLglobally \LTLeventually g_1$ which states that $g_1$ needs to appear infinitely often on a trace. We have that $\pi\vDash \varphi$ and $R = \mathbb{N}\times \ip$ serves as a \SR~(\Cref{obs:NNSR}). This is true for every $R' $ with infinitely many $g_1$'s. Thus, we can create an infinite chain of \SRs: $\mathbb{N} \times\set{i_1}\supsetneq\mathbb{N}_{>1}\times\set{i_1}\supsetneq\mathbb{N}_{>2} \times\set{i_1}\supsetneq \cdots $,
    where
    %, we use $[n]$ to denote the set $\{ 1, 2, \cdots, n\}$, and 
    $\mathbb{N}_{>n}=\{ m \in \mathbb{N} ~|~ m>n \}$.
    Since for every $j$ we have $\mathbb{N}_{>j} \supsetneq \mathbb{N}_{>j+1}$, there is no \MSR. 

%    As for minimum sufficient reasons, all sufficient reasons are countable, thus of the same power, therefore all sufficient reasons are minimum sufficient reasons. 
\end{example}

\subsection{Logical Encoding of Sufficient Reasons}\label{sec: Sufficient Reasons as Hyperproperties}
In this section, we show how to encode (M)\SRs~using HyperLTL. In \Cref{sec: complexities} we use this encoding to provide an algorithmic approach for the verification and synthesis of (M)SRs, via HyperLTL model checking. 

The set $R$ is a \SR~for $\phiTpi$ 
%for the fact that $\varphi$ holds on a trace $\pi$ in $T$, 
if the following formula holds:
\begin{align*} %\label{eq:reasonR}
   f[R, \pi]:=
    \forall \pi'\in\TR(T). \big( \forall \tuple{t,a}\in R. ~\pi'[t] =_{\{a\}} \pi[t] \big) \rightarrow \varphi_{\pi'}
\end{align*}
%However, the formula $f[R, \pi]$ is not a proper HyperLTL formula, as in HyperLTL we cannot quantify over time points. To overcome this, 
%
We can express $f[R, \pi]$ in HyperLTL using 
%we can use the auxiliary trace quantification proposed by HyperQPTL. 
%we use
$\tau_R$, the trace representation of $R$:
%(cf.~\Cref{def:tracerepresentation}): 
%and the additional proposition $\on$ to indicate the indices of $R$.  Intuitively, $R$ contains all indices $i$ on which $\on_{\tau_R}$ holds.
\begin{align*} %\label{eq:reasonTau}
 \psi_{\text{\SR}}[\tau_R, \pi]: =  \R_{\tau_R} \wedge  \forall \pi'. \Big( \neg \R_{\pi'} \wedge 
    \LTLglobally   \bigwedge_{a\in \ip}
    \big(   
    a_{\tau_R} \rightarrow (\pi =_{\{a\}} \pi')
    \big) \Big)
     \rightarrow \varphi_{\pi'}
\end{align*}
%
%The formula $\psi_{sr}[\tau_R, \pi]$ states that 
Meaning, if $\pi'$ agrees with $\pi$ on all inputs indicated by $R$, in the corresponding timesteps, %$\on$ is true on $\tau_R$, 
then $\pi'$ has to satisfy $\varphi$.  
HyperLTL implicitly quantifies over system traces so we don't need to add the fact that $\pi'$ is in $\TR(T)$ explicitly. Instead, we use $\neg\R_{\pi'}$ to indicate that it is not a trace representation, therefore, indeed one of the system traces. 

Note that if $a_{\tau_R}$ holds at timestep $t$, it does not mean that $a$ needs to hold on $\pi'$ at that timestep, but that $\pi$ and $\pi'$ should agree on the value of $a$: either ($a_\pi$ and $a_{\pi'}$) or  ($\neg a_\pi$ and $\neg a_{\pi'}$) at time $t$.

\begin{lemma}\label{lemma:RiffTau}
Let $R \subseteq \mathbb{N}\times \ip$ and $\tau_R $
be its trace representation. 
%, such that: $R = \{ i\in \mathbb{N}~|~ \on\in\tau_R[i]  \}$. 
Then, $f[R, \pi] $ 
holds in $T$
iff $ \psi_{\text{\SR}}[\tau_R, \pi]$ holds in~$T_\aug$. 
%\Cref{eq:reasonR} holds with respect to $\tuple{T, \pi}$ iff \Cref{eq:reasonTau} holds with respect to $\tuple{T_\aug, \pi}$. 
\end{lemma}

The proof is straightforward from the definition of 
$f[R, \pi] $, $ \psi_{\text{\SR}}[\tau_R, \pi]$ and $T_\aug$. %We provide it in the full version of this paper for completeness.
We provide a formal proof in Appendix~\ref{app:proofs} for completeness. 
Given \Cref{lemma:RiffTau}, we can use $ \psi_{\text{\SR}}[\tau_R, \pi]$ to verify or find sufficient reasons, as we show in the following. 

\subsubsection{Minimal Sufficient Reasons}\label{sec: minimal sufficient reasons}
To express the fact that $R$ is a \MSR, we only need to state that there is no \SR~$R' \subsetneq R$:
\begin{align}
     \psi_{\text{\MSR}}[\tau_R, \pi]: =  \psi_{\text{\SR}}[\tau_R, \pi] ~\wedge
    \Big( \forall \tau_{R'}. ~  \psi_{\text{\SR}}[\tau_{R'}, \pi] \rightarrow \big( (\tau_R \equiv \tau_{R'}) ~\vee \\ 
    %\neg\LTLglobally \big((\tau_{R'} = \pi) \rightarrow (\tau_R = \pi)\big)\big) \Big) 
      \neg\LTLglobally \bigwedge_{a\in\ip}\big(a_{\tau_{R'}} \rightarrow a_{\tau_R} \big)\big) \Big)    \label{eq:subsumed}
\end{align}
 \Cref{eq:subsumed} states that $R'$ is not subsumed by $R$, for any \SR~$R' \neq R$. 
%\nadav{We can write instead of 7:  $\exists i. \big( \tau_{R'}[i] \vDash  \text{on} \wedge \tau_R[i] \vDash \neg \text{on} \big)$}

\begin{example}

 Consider $\varphi = \LTLeventually \big(r_1\wedge \LTLnext (r_2\LTLuntil r_1)\big)$, the system of~\Cref{fig:motivation} and
    \begin{align*}
        \pi = \{r_1, g_1\}\{r_2\}\{r_1, g_1\}\{\}\{r_1, g_1\}\{r_2\}\{r_2\}\{r_1,g_1\}
        \{\}  \{r_1,g_1\}\{r_2\}\{r_2\}\{r_2\}\{r_1,g_1\}\{\}\ldots
        \end{align*}
    %Consider $\pi$, $T$, and $\varphi$ from \Cref{ex:min}. 
    There are infinitely many \MSRs, e.g. $  R_0 = \{ \tuple{0, r_1}, \tuple{1, r_2} ,\tuple{2,r_1}\}$, $ R_1 = \{ \tuple{4, r_1}, $ $\tuple{5, r_2} ,\tuple{6, r_2} ,\tuple{7,r_1}\}$, $ R_2 =  \{ \tuple{9, r_1}, \tuple{10, r_2} ,\tuple{11, r_2}, \tuple{12, r_2} ,\tuple{13,r_1}\} $,
    %
    %$R_0 = \{0,1,2\}, R_1 = \{4,5,6,7\}, R_2 = \{9,10,11,12,13\}$ 
    and so on. For each of these reasons $R_i$, any subset $R_i' \subsetneq R_i$ would not qualify as a SR. 
\end{example}

%We now show that if there is a (minimal) sufficient reason, then there exists a (minimal) sufficient reason that is \textbf{finitely representable} via an ultimately periodic trace. In particular, it means that we can always find an ultimately periodic trace that represents a (minimal) sufficient reason.  
We now show that we can restrict the search and verification of (M)\SRs~to only those that are finitely representable: 

\begin{theorem}\label{thm:omega:reg:reason}
    Let $\pi$ be an ultimately periodic trace, $T$ be a transition system, and $\varphi$ be an LTL property. Then, the following statements are equivalent: 
    \begin{enumerate}[nosep]
        \item Exists a (M)\SR~$R$ for $\phiTpi$.
        \item Exists an \emph{ultimately periodic} (M)\SR~$R'$ for $\phiTpi$.
        \item Exists an ultimately periodic trace representation $\tau_R$ of $R$ in 
        the augmented system $T_\aug$. \label{item:3}
    \end{enumerate}
\end{theorem}

\begin{proof} 
If there is a \SR~$R$, then due to \Cref{lemma:RiffTau}, 
 the formula $ \psi_{\text{\SR}}[\tau_R, \pi]$  holds in $T_\aug$
(and, similarly $   \psi_{\text{\MSR}}[\tau_R, \pi]$ for \MSRs). 
%, and the trace representation of $R$,$\tau_R$, is the witness of the respective hyperLTL formulas. 
That is, $\tau_R$ is the witness for the satisfaction of the HyperLTL formula $\exists \tau_R.  \psi_{\text{\SR}}[\tau_R, \pi] $. 
We use the fact that $\pi$ is ultimately periodic to explicitly encode it within the HyperLTL model-checking instance. 
If such a witness $\tau_R$ exists, there is also an ultimately periodic witness $\tau_{R'}$, since we find it using an emptiness check for the respective NBA~\cite{ClarksonFKMRS14}.
$\tau_R'$ corresponds to an ultimately periodic set $R'$ which also serves as a  (M)\SR. \qed
\end{proof}

\subsection{Complexities and Algorithms of \SRs~Verification and Synthesis}\label{sec: complexities}

In the following, we discuss the complexities of the verification and synthesis of \SRs, using
the HyperLTL encodings of \Cref{sec: Sufficient Reasons as Hyperproperties}.
We note that the upper-bounds also provide us with an algorithmic approach (which we implemented in our prototype) for the verification and synthesis of reasons. 
\subsubsection{Verification of Sufficient Reasons}

\begin{theorem}\label{thm:SRcomplexities}\label{thm:SRpspacecomplete}
     Given an LTL property $\varphi$, an ultimately periodic set $R$, an ultimately periodic trace $\pi$, and a system $T$, %, we have the following:
 %    \begin{enumerate}[nosep]
         %\item 
         the problem of verifying whether $R$ is a \SR~for $\phiTpi$ is PSPACE-complete. 
\end{theorem}

To prove \Cref{thm:SRpspacecomplete} we first prove the following Lemma: 

\begin{lemma}\label{lemma:SRemptyset}
    Let $T$ be a transition system, let $\varphi$ be an LTL formula, and let $\pi\in\TR(T)$ such that $\pi\models \varphi$. Then,
    $R = \emptyset$ is a sufficient reason for $\phiTpi$ iff $\varphi$ holds on all traces $\pi'\in\TR(T)$. 
\end{lemma}
\begin{proof}
    If $R = \emptyset$, the condition of \Cref{def:SR} holds vacuously, thus, we require for all traces $\pi'$ to satisfy $\varphi$. %, that is, all traces have to satisfy $\varphi$. 
    In the other direction, if all traces satisfy $\varphi$, then every set $R$ constitutes a sufficient reason, and in particular $R=\emptyset$. 
    
   Note that $\emptyset$ is finitely representable with trace representation $\tau_\emptyset = \{\R\}^\omega$.\qed
\end{proof}

\begin{proof}[of \Cref{thm:SRpspacecomplete}]
%\begin{enumerate}
    %\item  \label{item:in:pspace}
    We reduce (via polynomial time and space reduction) the problem of verifying \SRs~to LTL model-checking, which is PSPACE-complete~\cite{DBLP:journals/jacm/SistlaC85,DBLP:conf/popl/LichtensteinP85}.
    
    {Membership in PSPACE}:
   % The formula 
   Since $\pi$ and $\tau_R$ are fixed, and $\pi'$ is the only remaining variable, and is universally quantified, we can interpret $\psi_{sr}[\tau_R, \pi]$ as an LTL formula over traces $\pi'$ (see LTL semantics in \Cref{sec: prelim}).
    Since $\pi$ and $\tau_R$ are ultimately periodic, we can encode them as part of the instance, resulting in at most a polynomial blow-up. In addition, system $T_\aug$ is linear in the size of $T$, and the formula $\psi_{sr}[\tau_R, \pi]$ adds a constant factor to the size of $\varphi$. 
   Then, we simply need to model-check if $\psi_{sr}[\tau_R, \pi]$ holds in $T_\aug$.

   PSPACE hardness: We show a reduction from LTL model-checking: Given a system $T$ and an LTL formula $\varphi$, we construct an instance: $T'$, $\pi$, $\varphi'$ and $R$ in polynomial time, such that $T\models \varphi$ iff $R$ is a \SR~for $\tuple{\varphi', T', \pi}$.
    \begin{itemize}[nosep]
        \item We construct $T'$ from $T$ by adding a fresh state $q_\fresh \notin Q$, and the transitions 
        $\delta(q_0, \{ \fresh\})=( q_\fresh, \emptyset)= \delta(q_\fresh, \{ \fresh\})$
        %$(q_0, \{ \fresh\}, q_\fresh), (q_\fresh, \{ \fresh\}, q_\fresh)$ 
        for a fresh input proposition $\fresh \notin \ap$. 
     We denote $\pi_\fresh = \{\fresh\}^\omega$, and
    we have that that $\TR(T') = \TR(T) \cup \{ \pi_\fresh \}$.
    \item  We define $\varphi' := \varphi \vee \LTLglobally \fresh$. Then,
          $\pi_\fresh \vDash \varphi'$, and $\forall \pi' \in \TR(T)$: $\pi'\nvDash \LTLglobally \fresh$, therefore
         $\forall \pi' \in \TR(T)$: $\pi'\vDash \varphi $ iff $\pi'\vDash \varphi' $.
    % \item We choose $\pi = \pi_\fresh$.
     \item We choose $\pi = \pi_\fresh$, and we choose $R = \emptyset$, so the induced trace is $\tau_R = \emptyset^\omega$. 
    \end{itemize}
    The construction of $T'$ is linear in $T$ (both in terms of time and space), and so is the construction of $\varphi'$. $\pi_\fresh$ and $R$ are constant, so all together we have a polynomial construction. 

    We are left to show that 
     $T\models \varphi$ iff $\emptyset$ is a \SR~for $\tuple{\varphi',T', \pi_\fresh}$. Due to \Cref{lemma:SRemptyset}, %since $\pi_\fresh \vDash \varphi'$, 
     we have that 
     $\emptyset$ is a \SR~for $\tuple{\varphi', T', \pi_\fresh}$ iff $\varphi'$ holds on all traces in $\TR(T')$. Since $\varphi'$ holds on $\pi_\fresh$ regardless of $T$ and $\varphi$, and from the construction of $T'$, we have that $\varphi'$ holds on all traces in $\TR(T')$ iff $\varphi$ holds on all traces in $\TR(T)$, as needed. \qed
\end{proof}

For \MSRs~verification, we can use HyperLTL model-checking to check whether $ \psi_{\text{\MSR}}[\tau_R, \pi]$ holds in $T_\aug$. This implies at most double-exponential upper bound, due to the quantifier alternation introduced in $ \psi_{\text{\MSR}}[\tau_R, \pi]$ (the universal quantifier on the left side of the implication becomes an existential quantifier when transforming to prenex normal form). For hardness,  we can use the same reduction as in the proof of \Cref{thm:SRpspacecomplete}: $R = \emptyset$ is a~\SR~iff it is also a \MSR. 

\begin{corollary}
    The problem of verifying whether $R$ is a \MSR~is PSPACE-hard. 
\end{corollary}

\subsubsection{Synthesis of Minimal Sufficient Reasons}

We now consider the question of \emph{finding} \SRs, rather than verifying their correctness. 
%\nadav{but what if the LTS is non detrministic, meaning with the same input we can get different traces which might not satisfy $\varphi$.}
%
%
      From \Cref{thm:NNSR} it follows that there always exists a \SR, and we can choose it to be $\mathbb{N}\times \ip$. Thus, the question of synthesizing an arbitrary \SR~becomes trivial. 
      Therefore, in the following, we are interested in synthesizing \emph{\MSRs}.
      %- we can always choose $R = \mathbb{N}$. 
  
\begin{theorem}\label{thm:SRexptime}
    We can reduce the synthesis problem of \MSRs~to the problem of HyperLTL model checking. In particular, the synthesis problem of \MSRs~is at most $3$EXPTIME in the size of $\varphi$. 
\end{theorem}
\begin{proof}
    Consider the following HyperLTL formula $ \psi_{ml}[\pi]  :=   \exists \tau_R. \psi_{\text{\MSR}}[\tau_R, \pi] $. 
    The formula $\psi_{ml}[\pi] $ holds in a transition system $T$ iff there exists $\tau_R$ such that $R$ serves as a \MSR~for $\phiTpi$.
    That is, we have reduced the synthesis problem of \MSRs~to the problem of model-checking $\psi_{ml}[\pi] $ on $T$. Since $\psi_{ml}[\pi] $ %is an $\exists \forall \exists$ formula, it 
    contains two quantifier alternations, its model checking is double-exponential~\cite{finkbeiner2015algorithms}. \qed
\end{proof}

As we demonstrate in \Cref{sec:experiments}, this upper-bound doesn't imply practical hardness, and we manage to find minimal sufficient reasons in a reasonable time.

\section{Temporal Sufficient Reasons}\label{sec:TemporalReasons}

\Cref{ex:infinite} shows that in some cases we cannot find useful \SRs~as set of events on the trace. To overcome this, in the following, we look for temporal properties which will serve as reasons, and show that %in the case of temporal reasons, 
there always exists a unique minimal temporal \SR, which can be represented as an NBA.
As before, we assume the effect is given as an LTL formula $\varphi$.\footnote{All described algorithms, both of \Cref{sec:nontemporal} and of this section, work also if the effect is given directly as an NBA: We anyway translate the LTL formula into an NBA to solve the model-checking problem. If we assume an NBA input, the complexity results of \Cref{sec:nontemporal} are even better, as the PSPACE hardness partially stems from the exponential translation of the LTL formula into an NBA.\label{note:nbaVSltl}}
\begin{definition}[Temporal Sufficient Reasons]\label{def:TSR}
Let $T$ be a system, 
 $\pi \in \TR(T)$, $\varphi$ an LTL property over $\ap$ with $\pi\vDash\varphi$, and $P_R$ a trace property~over~$\ip$. 
%  Let $\pi \in \Sigma^\omega$ be a trace, let $\varphi_{\textit{spec}}$ be a specification, and let $\varphi_{\textit{reas}}$ be an LTL formula.
  \begin{itemize}[nosep]
      \item We say that $P_R$ %over the set of atomic propositions $\ip$ is a 
    is a 
    \emph{temporal sufficient reason (\TSR)} for
      $\phiTpi$ if $\pi $ satisfies $ P_R$ and 
      %$f_{\varphi_{\textit{spec}}}(\pi) = b$ 
      if for every trace $\tau\in \TR(T)$, if $\tau$  satisfies $P_R$ then $\tau \vDash \varphi$. This is exactly the temporal extension of \Cref{def:SR}: Since we assume $\pi \vDash \varphi$, we need to check that all traces that agree with $\pi$ on the satisfaction of $P_R$, would also satisfy $\varphi$. 
      \item  We say that $P_R$ is a \emph{minimal temporal sufficient reason (\MTSR)} if there is no
      %for every 
      \TSR~$P_{R'}$ such that $P_{R'} \supsetneq P_R$.
     % $\varphi_{R'}$ we have that $\varphi_{R'} \rightarrow \varphi_{R}$.
     %$\varphi_{\textit{reas}} \rightarrow \psi$. %\textcolor{red}{is this the correct direction of implication? maybe not as this is the opposite to the size criterion}
  \end{itemize}
\end{definition}

\begin{remark}
In \Cref{def:SR} $R$ is minimal if there is no $R'\subsetneq R$, while in \Cref{def:TSR} $P_R$ is minimal if there is no $P_{R'} \supsetneq P_R$. 
These follow the same reasoning, and are not contradicting:  
     Unlike other approaches for temporal causality, which look for the strongest condition that implies the effect~\cite{DBLP:conf/atva/CoenenFFHMS22,FFMS24}, here we look for the weakest such condition. 
      $P_R = (2^\ip)^\omega$ subsumes all trace properties over $\ip$, 
      %\emph{true} is implied by all other LTL formulas, 
      and will be a \MTSR~only if all traces satisfy the effect $\varphi$,  as in the case of $R = \emptyset$.
        
          %Note that the trace property $P_R$ can, in particular, be the language of some LTL formula, as we show in the following example.  
      
\end{remark}

\begin{example}\label{ex:temporal}
    Consider again the trace  $\pi = \{r_1, g_1\}^\omega$ and the effect $\varphi = \LTLglobally \LTLeventually g_1$  of \Cref{ex:infinite}. Then 
    the following are examples of possible \TSRs~for $\phiTpi$: $P_{R_1} = \mathcal{L}(\LTLglobally r_1)$, and $P_{R_2} = \mathcal{L}(\LTLglobally \LTLeventually r_1)$. In this case, $P_{R_2}$ also serves as a \MTSR. 

    This demonstrates the need in temporal reasons: Recall that due to the infinite chain of decreasing \SRs, we could not find a \MSR~for the effect $\LTLglobally\LTLeventually g_1$ (\Cref{ex:infinite}). The use of symbolic temporal reasons overcomes this difficulty. 
%    ... \textcolor{red}{this looks like we will always get the spec itself, which is very not useful. Then maybe we do want to talk about input-output and the dynamics of the system? in any case this can be a motivation for why introducing the additional level of input-output and systems. ok, it will not always be the spec, as we now show:}

    %\textcolor{red}{if we use the same implication, we get that any specific occurrence is a minimal cause, that is $\LTLnext^i$. but this was actually the case before as well, without introducing temporal causes. So........ do we really need them? or should we rephrase to "the un-usefulness of temporal causes"? }
\end{example}

\subsection{Verification of Temporal Sufficient Reasons} 

As before, we first show how to encode \TSRs~in HyperLTL, and then we use HyperLTL model checking to verify them. 
$P_R$ is a \TSR~for $\phiTpi$ if %$ \forall \pi'\in\TR(T). \pi'_{\downarrow_\ip} \in P_R \rightarrow \pi' \vDash \varphi$. 
\begin{align}\label{eq:buchiP}
    \forall \pi'\in\TR(T). \pi'_{\downarrow_\ip} \in P_R \rightarrow \pi' \vDash \varphi
\end{align}

If $P_R$ is the language of an LTL formula $\varphi_R$, we have that $\varphi_R$ is a \TSR~if:% $\psi_\text{tsr}$ holds, where $\psi_\text{tsr} =   \forall \pi'. {\varphi_R}_{\pi'} \rightarrow \varphi_{\pi'}$. 
\begin{align}\label{eq:TSR}
    \forall \pi'. {\varphi_R}_{\pi'} \rightarrow \varphi_{\pi'}
\end{align}
\Cref{eq:TSR} is equivalent to the LTL formula ${\varphi_R} \rightarrow \varphi$, since LTL implicitly universally quantifies over all system traces. The implication has to be true only for $\TR(T)$, and does not have to hold for all possible traces.

\begin{theorem}\label{thm:pspacetemporal}
    Given a trace property $P_R$ as an NBA 
    or an LTL formula, a system $T$, an ultimately periodic trace $\pi\in\TR(T)$, and an LTL effect $\varphi$, 
    verifying if $P_R$ is a \TSR~for $\phiTpi$ is PSPACE-complete. 
    %This holds also if $P_R$ is the language of an LTL formula. 
%
\end{theorem}

\begin{proof}[sketch] 
    Membership in PSPACE: 
First, assume that $P_R$ is given as an NBA over $2^\ip$.
    Let $\mathcal{A}_{P_R}$ be the extension of $P_R$ to the full set $2^\ap$ (c.f. \Cref{def:extension}).
    Then,   \Cref{eq:buchiP} reduces to model-checking whether $\mathcal{A}_{P_R} \cap T \vDash\mathcal{A}_\varphi$.

    If $P_R$ is an LTL formula, to verify that $P_R$ is a \TSR~for $\phiTpi$ we need to model-check \Cref{eq:TSR} against $T$, which is an LTL model-checking problem.

    PSPACE hardness: We reduce LTL model-checking to \TSRs~verification, with almost the same proof as for \Cref{thm:SRpspacecomplete}, but now we are given the trace property $P_R = \mathit{true}$ rather then the empty set of events (an automaton that accepts all words has one state, and an LTL formula is of constant size):
       $\mathit{true}$ is a \TSR~for $\phiTpi$ iff all traces of $T$ satisfy $\varphi$. 
       %See the full version of this paper for the full proof. 
       We provide the full proof in Appendix~\ref{app:proofs}.  
       \qed
       %for completeness. 
\end{proof}

\subsubsection*{Verification of Minimal Temporal Sufficient Reasons}
%For minimal sufficient reasons, as we show 
In \Cref{sec:temporalsynthesis} we show that there is a unique \MTSR~$S$ given $\varphi$, $T$ and $\pi$. Therefore, to verify whether a traces property $P_R$
is a \MTSR, 
%, given as a Büchi automaton or as an LTL formula, 
all we need to do is to check equivalence between $P_R$ and the set $S$ given in \Cref{sec:temporalsynthesis}.

\subsection{Synthesis of Minimal Temporal Sufficient Reasons}\label{sec:temporalsynthesis}
In the following, given an LTL effect $\varphi$, we show an automata construction for the \MTSR~for $\phiTpi$. Even when the effect is given as an LTL formula, the \MTSR~might not be expressible in LTL, and require the $\omega$-regular power of NBAs. This was also demonstrated in~\cite{DBLP:conf/atva/CoenenFFHMS22} for Halpern and Pearl's actual causality. We observe here that this is a result of the $\omega$-regular nature of the transition system $T$, and not due to the specific causality definition.

\begin{example}\label{ex:omegareg}
    Let $T$ be the transition system of \Cref{fig:qptl}, and consider the trace $\pi = \set{i_1, o_1}^\omega$, and the effect $\LTLglobally o_1$. 
    Possible \TSRs~for $\LTLglobally o_1$ in $T$ are the LTL formulas $\LTLglobally i_1$ and $i_1 \wedge \LTLnext\LTLnext \LTLglobally i_1 $. However, since the input at every odd position does not influence the output, a \MTSR~only needs to constraint the inputs at even timepoints. In fact, the \MTSR~here for $\LTLglobally o_1$ is $P_R:=$\emph{\quotes{$i_1$ holds at every even timestep}}. First, $\pi$ satisfies $P_R$. Second, every trace that satisfies $P_R$ also satisfies the effect. Last, $P_R$ is minimal since every trace on which at some even position $i_1$ does not hold, will loop on $q_0$ and will output $\{\}$ at that point. 
    The property $P_R$ cannot be expressed in LTL, but there is a simple NBA for it. 
\end{example}

We now prove an automata characterization and construction of the \MTSR. For the proof, we use the following observation:

\begin{observation}
    \label{lemma:tracepropertyres}
    Let $S$ be a trace property over $\ip$. For $\pi_1,\pi_2\in (2^\ap)^\omega $, if $\pi_1 \equiv_\ip \pi_2$ then $\pi_1 $ satisfies $S$ iff $\pi_2$ satisfies $S$.
\end{observation}

\begin{theorem}\label{thm:minimaltraceproperty}
    Given a system $T$, $\pi\in\TR(T)$ and an LTL effect $\varphi$, 
    The \MTSR~for $\phiTpi$ is $\mathcal{L}((T \cap \mathcal{A}_\varphi)_{\downarrow_\ip})$, that is, the projection of $(T \cap \mathcal{A}_\varphi)$ over the set $\ip$.
\end{theorem}

%\begin{proof}
 %   We have ${\pi_1}_{\downarrow_{\ip}} = {\pi_2}_{\downarrow_{\ip}} $, thus    ${\pi_1}_{\downarrow_{\ip}} \in S \leftrightarrow {\pi_2}_{\downarrow_{\ip}} \in S$   and the lemma follows directly. 
%\end{proof}

\begin{proof}[of \Cref{thm:minimaltraceproperty}]
    We denote $S = \mathcal{L}((T \cap \mathcal{A}_\varphi)_{\downarrow_\ip})$.
    %\begin{itemize}[nosep]
         We first show that
     $S$
    is a \TSR: $\pi \in T \cap \aut_\varphi$, therefore $\pi_{\downarrow_\ip} \in S$, thus $\pi$ satisfies $S$.   

    Let $\pi' \in T $. We show that if $\pi'$ satisfies $S$ then $\pi'\vDash \varphi$. Denote $\pi' = \pi'_{\downarrow_\ip} \cup \pi'_{\downarrow_\op}$. Since $\pi'$ 
    satisfies $S$ we have $\pi'_{\downarrow_\ip} \in S$, and, by the definition of the projection, it means that there is $\pi'_o$ such that $\pi'_{\downarrow_\ip} \cup \pi'_o \in (T \cap \aut_\varphi)$. In particular, we have $\pi'_{\downarrow_\ip} \cup \pi'_o \in T$, and, by the assumption, we also have $\pi' = \pi'_{\downarrow_\ip} \cup \pi'_{\downarrow_\op} \in T$. Since $T$ is input-deterministic, $\pi' = \pi'_{\downarrow_\ip} \cup \pi'_o $, therefore $\pi' \in \aut_\varphi$ as needed.

Minimality: Let $R$ be a \TSR, we show that $R \subseteq S$. Let 
    $\pi'$ satisfy $R$, that is 
    $\pi'_{\downarrow_\ip}\in R$. 
    We need to show that    $\pi'_{\downarrow_\ip}\in S$. 
    %such that $\pi' = \pi'_{\downarrow_\ip} \cup \pi'_{\downarrow_\op}$. 
    We distinguish between two cases: 
    \begin{itemize}[nosep]
        \item $\pi' \in T$: Since $R$ is a \TSR~for $\varphi$, we have that $\pi'\vDash \varphi$, therefore $\pi' \in T\cap \aut_\varphi$, which means that    $\pi'_{\downarrow_\ip}\in S$ as needed. 
        \item $\pi' \notin T$: Denote    $\pi' = \pi'_{\downarrow_\ip}\cup \pi'_{\downarrow_\op}$. Since $T$ is input deterministic, there is $\pi'_{o}$ such that $\pi'' = \pi'_{\downarrow_\ip} \cup \pi'_{o}$ and $\pi'' \in T$. Due to \Cref{lemma:tracepropertyres}, since $\pi'\equiv_\ip \pi''$, and since $\pi'$ satisfies $R$, we have that $\pi''$ satisfies $R$. Therefore, since $R$ is a \TSR, $\pi'' \in \aut_\varphi$, and thus $\pi'' \in T\cap \aut_\varphi$, which means that $\pi''$ satisfies $S$. From \Cref{lemma:tracepropertyres}, and since $\pi'\equiv_\ip \pi''$, we conclude that $\pi'$ satisfies $S$. \qed
    \end{itemize}
\end{proof}

\begin{corollary} Given a transition system $T$, ultimately periodic $\pi$, and an effect~$\varphi$,
    the \MTSR~for
    $\phiTpi$ is $\omega$-regular and {unique}.\footnote{
    Here, $\varphi$ can be an LTL formula or some $\omega$-regular trace property, see \Cref{note:nbaVSltl}.}
\end{corollary}

\begin{remark}
    The \MTSR~is independent of $\pi$: The only constraint is that $\pi \vDash \varphi$. Our framework allows one to choose the non-temporal and more concrete reasons, or the temporal more general ones.  
\end{remark}
%
%To summarize, we have the following:
\begin{theorem}\label{thm:temporalReasonscomplexities}
Given $T, \pi$ and an LTL effect $\varphi$,
     finding the \MTSR~for $\phiTpi$ is exponential in $\varphi$. 
\end{theorem}

\begin{proof}
    The proof follows directly from the construction of the \MTSR~in \Cref{thm:minimaltraceproperty}, and the fact that $\mathcal{A}_\varphi$ is exponential in $\varphi$ \cite{DBLP:conf/lics/VardiW86}. \qed
\end{proof}

\begin{figure}[t]
    \centering
   % \hspace{0.1cm}
    %
    \begin{minipage}[t]{0.37\textwidth}
        \centering
        \begin{tikzpicture}[scale=0.13]
            \tikzstyle{every node}+=[inner sep=0pt]
            \draw [black] (27.4,-23.4) circle (3);
            \draw (27.4,-23.4) node {$q_0$};
            \draw [black] (43.9,-23.4) circle (3);
            \draw (43.9,-23.4) node {$q_1$};
            \draw [black] (18.4,-23.4) -- (24.4,-23.4);
            \fill [black] (24.4,-23.4) -- (23.6,-22.9) -- (23.6,-23.9);
            \draw [black] (29.843,-21.672) arc (118.2625:61.7375:12.263);
            \fill [black] (41.46,-21.67) -- (40.99,-20.85) -- (40.52,-21.73);
            \draw (35.65,-19.71) node [above] {$\{i_1\}|\mbox{ }\{o_1\}$};
            \draw [black] (41.757,-25.485) arc (-54.05228:-125.94772:10.403);
            \fill [black] (29.54,-25.48) -- (29.9,-26.36) -- (30.48,-25.55);
            \draw (35.65,-27.97) node [below] {$\{\},\{i_1\}|\{o_1\}$};
            \draw [black] (25.266,-21.308) arc (253.29005:-34.70995:2.25);
            \draw (23.77,-16.41) node [above] {$\{\}|\{\}$};
            \fill [black] (27.76,-20.43) -- (28.47,-19.81) -- (27.51,-19.52);
        \end{tikzpicture}
        \caption{A system over $\ip = \{i_1\}, \op = \{o_1\}$ for \Cref{ex:omegareg}.}
        \label{fig:qptl}
    \end{minipage}
    \hfill
    \begin{minipage}[t]{0.59\textwidth}
         \includegraphics[width=\linewidth]{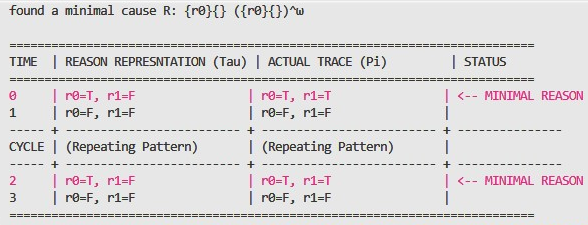}
    
    % The caption will be confined to the 0.5\textwidth width
    \caption{\textbf{\MSR~synthesis.} Above CYCLE is the prefix of the trace, while below CYCLE is the $\omega$-repeating suffix.}
    \label{fig:msr_output}
    \end{minipage}
\end{figure}

\section{Contrastive Explanations}
\label{sec:ContrastiveExplanations}
Contrastive explanations are system traces that, in contrast to $\pi$, do not satisfy the effect. We will usually be interested in \emph{minimal} contrastive explanations, i.e., the ones that are the closest to the trace $\pi$, to understand what are the minimal changes needed to be done in order to avoid the effect. %In addition, while the definition of~\cite{DBLP:journals/pacmmod/BarceloKRSV25}, as given in \Cref{sec: prelim} only req 
To reason about such minimal explanations formally, we first define a distance between traces. 
%
%\paragraph{Distance between traces}
We define a distance measure $\Delta:\alltraces \times \alltraces \rightarrow 2^{\mathbb{N}\times \ip}$ to be the set of all inputs and timesteps on which the two traces disagree:
         \begin{align*}
             \Delta(\pi_1, \pi_2) = \{ (n, i) ~|~
         n\in\nat, i\in\ip,
         i\in \pi_1[n] \setminus \pi_2[n] \vee  i\in \pi_2[n] \setminus \pi_1[n]\}
         \end{align*}
\begin{definition}[Contrastive Explanations]
    Let $T$ be a transition system, let
 $\pi,\cont \in \TR(T)$, and let $\varphi$ be an LTL formula such that $\pi\vDash\varphi$.
    \begin{itemize}[nosep]
        \item We say that $\cont$ is a \emph{contrastive explanation (\CE)} for 
        $\phiTpi$ 
        %$\pi \vDash \varphi$ 
        if $\cont \nvDash \varphi$. 
        \item 
        We say that $\cont$ is a \emph{minimal contrastive explanation (\MCE)} for $\phiTpi$ if $\cont \nvDash \varphi$, and, for every trace $\tau'\neq \cont\in \TR(T)$ such that $\tau'\nvDash \varphi$ we have $\Delta(\pi, \tau') \not\subseteq \Delta(\pi, \cont)  $. %\footnote{This is the common approach when reasoning about temporal causality and counterfactuals, see e.g. \cite{FFMS24}.} 
        That is, there is no \CE~that is closer to $\pi$ than $\cont$. 
    \end{itemize}
\end{definition}
\begin{remark}\label{remark:ces as set if events}
    We note that \CEs~are traces, and thus are not of the same type as \SRs, which are sets of event. This choice follows the definition of \cite{DBLP:journals/pacmmod/BarceloKRSV25} that defines a \CE~as another input vector, and \SRs~as a subset of components of the input vector. However, one can easily convert a \CE~$\cont$ to be in the form of sets of events, by taking $\Delta(\pi,\cont)$.
\end{remark}

\subsection{Verification and Synthesis of 
\CEs~using HyperLTL Encoding}\label{sec:CE via hyperltl}

As we did for \SRs, we first show how to encode \CEs~in HyperLTL, and then~use this encoding for verification and synthesis of \CEs. Since finding some trace~that does not satisfy the effect is simply a model checking problem, we focus on~\MCEs.

%\paragraph{Expressing contrastive explanations in HyperLTL}
The trace $\cont$ is a \MCE~for $\phiTpi$ if any trace $\pi'$ that agrees with $\pi$ whenever $\cont$ agrees with $\pi$, that is, $\Delta(\pi, \pi') \subseteq \Delta(\pi,\cont)$, satisfies $\varphi$. 
We formalize this using the formula
$\psi_{\text{\CE}} [\cont]$:
\begin{align*} 
    \psi_{\text{\CE}} [\cont]:= \neg \varphi_{\cont} \wedge
    %\neg \varphi_{\cont} \wedge
   % \label{eq:count1}
     \forall\pi'. (\pi' \not\equiv \cont) \rightarrow
    \Big(\LTLglobally  \bigwedge_{p\in \ip} \big((p_{\cont} \leftrightarrow p_\pi) \rightarrow  (p_{\pi'} \leftrightarrow p_\pi)\big) \rightarrow \varphi_{\pi'} \Big) %\label{eq:count2}
\end{align*}
%\nadav{if I understood 12 correctly what if we take $\pi' =\cont $}
%Counterfactuals were studied in previous works~\cite{julian...}, as a mechanism to express other notions of explainability. Here, we are interested in generating the counterfactuals themselves, and in the complexity of verifying or finding counterfactual explanations, which was not studied before. 

\subsubsection{Verification of Minimal Contrastive Explanations}

\begin{theorem}\label{thm:cont}Given $T$, $\pi$, $\varphi$, and an ultimately periodic trace $\cont$, 
    the problem of verifying if $\cont$ is a \MCE~for
     $\phiTpi$
is PSPACE-complete. 
\end{theorem}

\begin{proof}
    {Membership in PSPACE}: 
    We can verify in polynomial time that $\cont \nvDash \varphi$ \cite{DBLP:conf/concur/MarkeyS03}. 
If   $\cont \nvDash \varphi$   we can explicitly encode the traces $\cont$ and $\pi$, 
    and then model-check $T$ against the rest of the formula
    $\psi_{\text{\CE}} [\cont]$.
    %specified in Equations~(\ref{eq:count1})-(\ref{eq:count2}).
    We note that if $\cont$ is fixed, then %the formula in  Equations~(\ref{eq:count1})-(\ref{eq:count2})
    $\psi_{\text{\CE}} [\cont]$
    is in fact an LTL formula, as LTL implicitly has a universal quantifier over all traces (here, the traces $\pi'$).
    %as we can verify it by model-checking the above formula on $T$. 

    PSPACE completeness: We reduce from LTL model-checking. Let $T$ be a transition system. 
    We construct a system $T'$, traces $\pi$ and $\cont$, and an LTL formula $\varphi'$ in polynomial time, such that $T\vDash \varphi$ iff $\cont$  is a \MCE~for $\tuple{\varphi',T', \pi}$: 
    \begin{itemize}[nosep]
        \item 
        Similar to the proof of \Cref{thm:SRpspacecomplete}, we construct $T'$ from $T$ by adding two fresh states $q_{\fresh_1},q_{\fresh_2} \notin Q$, and the transitions 
        $\delta(q_0, \{ \fresh_i\})=(q_{\fresh_i},\emptyset)$,
        $\delta(q_{\fresh_i}, \{ \fresh_i\}) =  (q_{\fresh_i}, \emptyset)$ for $i=1,2$ and fresh input propositions $\fresh_1, \fresh_2 \notin \ap$, so we set $\ip' = \ip \cup \{\fresh_1, \fresh_2 \}$. 
        We have that $\TR(T') = \TR(T) \cup \{ \{\fresh_1\}^\omega, \{\fresh_2\}^\omega \}$.
        \item We set $\varphi' = (\varphi \wedge \LTLeventually \neg \fresh_2) \vee \LTLglobally \fresh_1 $, $\cont = \{ \fresh_2\}^\omega$ and  $\pi =  \{ \fresh_1\}^\omega$.
    \end{itemize}
    
    Then, we have: 
    $\pi \vDash \varphi'$ and $\cont \nvDash \varphi'$; 
      for every $\pi' \in \TR(T)$, $\pi'\vDash \varphi'$ iff $\pi'\vDash \varphi$; 
       and $\cont$ is always different than $\pi$, so the left side of the implication is always evaluated to \emph{true}, and $\cont$ is a \MCE~iff $\varphi$ holds on all system traces. \qed
\end{proof}

\subsubsection{Synthesis of Minimal Contrastive Explanations}\label{sec:coountsynt}
To find a \MCE, we existentially quantify over $\psi_{\text{\CE}} [\cont]$, and
 model-check~$T$ against the formula $\exists \cont. ~\psi_{\text{\CE}} [\cont]$. 
 A witness to this formula is a \MCE.
% If there is a \MCE, it will be a witness to this formula.
%
%
%
Note that \CEs~are over the whole set $\ap$, but we can easily take the projection of such a witness over $\ip$ to get a \CE~over inputs only. 
The complexity is inherited from corresponding HyperLTL model-checking instance, and is at most double-exponential in~$\varphi$.

\subsection{Temporal Contrastive Explanations}
Last, we shortly discuss a symbolic representation of \CEs. 

\begin{definition}[Temporal Contrastive Explanations]\label{def:temporalcont}
Let $T$ be a transition system, let
 $\pi \in \TR(T)$, and let $\varphi$ be an LTL formula such that $\pi\vDash\varphi$.
 \begin{itemize}[nosep]
     \item A \emph{temporal contrastive explanation (\TCE)} for $\phiTpi$ is a set $S \subseteq (2^\ip)^\omega$ such that $S\subseteq (T \cap \mathcal{A}_{\neg \varphi}){\downarrow_\ip}$. 
     \item The \emph{minimal temporal contrastive explanation (\MTCE)} is 
   $S\subseteq (T \cap \mathcal{A}_{\neg \varphi}){\downarrow_\ip}$, such that
   for all traces $\tau \in \TR(T)$, we have that $\tau$ is a \MCE~iff $\tau_{\downarrow_\ip} \in S$. That is, $S$ is the set of all \MCEs, projected over $\ip$.
   %all traces in $S$ are minimal contrastive explanations. 
 \end{itemize}
\end{definition}

\subsubsection{Verification of Temporal Contrastive Explanations}
Given an $\omega$-regular set $S$ as an NBA over $\ip$, we can verify if $S$ is a \MTCE~by model-checking a slight modification of the formula $\psi_{\text{\CE}}$ above, as follows:
\begin{align*} %\label{eq:countver1}
    \forall {\cont} \in S. \forall \pi' \in T.~
\neg \varphi_{\cont} \wedge \pi' \not\equiv \cont \rightarrow
    \Big(\LTLglobally  \bigwedge_{p\in \ip} \big((p_{\cont} \leftrightarrow p_\pi) \rightarrow  (p_\pi' \leftrightarrow p_\pi)\big) \rightarrow \varphi_{\pi'} \Big) %\label{eq:countver2}
\end{align*}

%This requires model-checking over the two automata, $T$ and $S$, but this does not change the complexity of the problem\hadar{cite}. 

\begin{theorem}
    Given a set $S$, system $T$, trace $\pi$ and effect $\varphi$, verifying whether  $S$ is a \MTCE~for $\phiTpi$ is PSPACE-complete. 
\end{theorem}

\begin{proof}
    The two universal quantifiers only add a polynomial factor to the LTL model-checking problem, which is in PSPACE. 
    For PSPACE-hardness we can use the same proof as of \Cref{thm:cont}, with the set $S= \set{\cont}$.  \qed
\end{proof}

\subsection{Between Sufficient Reasons and Contrastive Explanations}

Previous works in the context of neural networks have established duality between \MSRs~and \MCEs~\cite{DBLP:conf/aiia/IgnatievNA020}. In the context of explaining temporal effects in reactive systems, this does not hold; while the notions of \MSRs~and \MCEs~are intuitively dual, we cannot directly derive, using a duality relation, one from the other. %For example, consider the trace $\pi = {}$
It is not enough to negate the events in a \MSR~to avoid the effect, since \MSRs~are not unique. 
One might suggest to negate \emph{all} \MSRs, however, this is too strict: it is enough to negate one event from each \MSR, to find a \CE. %
We formalize this 
%In the following, we establish the connection 
in \Cref{thm: msrs vs ces}. The proof is straightforward and appears in Appendix~\ref{app:proofs} for completeness.
Therefore, \MSRs~and \MCEs~are clearly related, but are not dual. The reason that duality fails here, while holds for the neural networks domain, is the nature of our problem: 
In other domains, such as~\cite{DBLP:conf/aiia/IgnatievNA020}, the only component needed for evaluation is the function $f$ which is explained. In our work, the function
we evaluate is the satisfaction of $\varphi$, but we have a second component that influences the decision, which is the system $T$ from which the traces are taken. 
%We explain the satisfaction of an LTL effect $\varphi$ on a trace, with respect to a system $T$. That is, the system determines the other possible traces with which $\pi$ is compared, however, the \emph{function} we evaluate is the satisfaction of $\varphi$. In other works, such as~\cite{DBLP:conf/aiia/IgnatievNA020}, both the evaluation, and the different inputs, are over the same function~$f$. 
%
\begin{theorem}\label{thm: msrs vs ces}
Given $T$, $\pi$ and $\varphi$, 
%Let $\pi$ be an ultimately periodic trace, let $\varphi$ be an LTL effect and let $T$ be a transition system. 
let $\mathfrak{R} = \{ R_i ~|~ i\in\mathfrak{I} \}$ be the set of all \MSRs~for $\phiTpi$. Then, for any \MCE~$\tau$ we have  
    $\forall i\in \mathfrak{I}: \Delta(\tau,\pi) \cap R_i \neq \emptyset$. 
%   
   % \item Let $S \subseteq \ip\times\nat$ such that   $\forall i\in \mathfrak{I}: S \cap R_i \neq \emptyset$. Then, there is a trace $\tau$ such that $\Delta(\tau,\pi) = S$ and such that $\tau$ is an \MCE~for $\phiTpi$.
%  
\end{theorem}

\section{Evaluation}
\label{sec:experiments}
In this section, we provide a brief evaluation of our approach using a prototype implementation. 
%
%
%In this section, we evaluate the performance and correctness of our proposed approach for verifying and synthesizing sufficient reasons (both temporal and non temporal), as well as synthesizing minimal contrastive explanation. %To ensure a fair and understandable comparison, we utilize benchmark examples derived from the CATS \cite{DBLP:conf/atva/BeutnerFFS23}.
%Our evaluation relies on a proof-of-concept implementation of 
%the algorithms described in Sections~\ref{sec:nontemporal}-\ref{sec:ContrastiveExplanations}. 
Our preliminary implementation meant to demonstrate our approach, but this is not yet an automated tool. In particular, 
for the purpose of these experiments, the specific transition systems, input traces, and effects were statically defined (hard-coded) within each test. The causal analysis using HyperLTL model checking remains general. %This design choice facilitates consistent benchmarking and isolates algorithmic performance from parsing overhead; however, 
%The core verification and synthesis logic remains general. %We compare our results against existing baselines where applicable, specifically the CATS tool~\cite{DBLP:conf/atva/BeutnerFFS23} and the CORP tool~\cite{FFMS24}.
In the following, we provide brief comparisons. %For a more detailed evaluation, 
Additional tables appear in Appendix~\ref{app:tables}. 

We report the results of our prototype, and, when possible, a comparison to existing tools for \emph{actual} temporal causality. This comparison is meant to give intuition about how the different causality and explainability approaches (e.g. \TSR~vs. actual causality) translate to different causes and explanations. The comparisons are not intended as a way to evaluate which approach is better, since, as we discuss in the introduction, each approach addresses a different notion of causality. However, since there are no other tools for {temporal} causality in reactive systems, except for actual causality, we find the comparison interesting. 
% (see supplementary material).  

%
%\textbf{Setup and Reproducibility}
All experiments were conducted on a machine running Windows Subsystem for Linux (WSL, Kernel 4.4.0) with an {Intel Core i7-8565U CPU @ 1.80GHz} and {16 GB of RAM}. %The reported execution times are the average of 5 runs for Table~\ref{tab:cats_results1}, and 100 runs for Table~\ref{tab:crop_vs_mtsr} to ensure statistical stability. 
Our prototype is available at the anonymous repository \href{https://github.com/SubmissionForPaper1/research-artifact.git}{here}.\footnote{\url{https://github.com/SubmissionForPaper1/research-artifact.git}.}

Table~\ref{tab:cats_results1} compares %performance of 
our \TSR~verification approach against the CATS tool for actual causality~\cite{DBLP:conf/atva/BeutnerFFS23}, using CATS benchmarks. %The evaluation was conducted using a set of benchmarks derived from the CATS artifact.
The fact that 
%\subsubsection{Performance Analysis}
our approach outperforms CATS in terms of runtime stems from the
fact that \TSRs~are much easier to verify than temporal causes defined by actual causality.
%nature of our verification task: verifying that a candidate formula is a Temporal Sufficient Reason corresponds to checking a simple hyperproperty (as formalized in Eq.\eqref{eq:TSR}). Unlike CATS ~\cite{DBLP:conf/atva/BeutnerFFS23}, which must explore the solution space to verify the causal relation, our implementation verifies the sufficiency of a provided reason directly against the system model. 

%\input{table2}
%\input{table2-new}

%\paragraph{Synthesizing Minimal Temporal Sufficient Reasons}
\begin{comment}
    
\begin{wrapfigure}{r}{0.5\textwidth} % {r} for right side, 0.5 for half-page width
    \centering
   % \vspace{-15pt} % Removes excess whitespace above the figure
    \includegraphics[width=\linewidth]{figures/MSR_picture_inverted3.png}
    
    % The caption will be confined to the 0.5\textwidth width
    \caption{\textbf{\MSR~synthesis.} The part above CYCLE refers to the prefix of the trace, while the part below CYCLE refers to the $\omega$-repeating suffix.}
    \label{fig:msr_output}
    \vspace{-10pt} % Removes excess whitespace below the caption
\end{wrapfigure}
\end{comment}

In \Cref{tab:crop_vs_mtsr} we evaluate the synthesis of \MTSRs, as defined in \Cref{thm:minimaltraceproperty}, and compare it with the CORP tool~\cite{FFMS24}, which synthesizes temporal actual causes.

\Cref{fig:msr_output} presents the output of our prototype for \MSRs~synthesis, 
based on \Cref{sec: minimal sufficient reasons}. 
In \cite{coenen2022hyperproperty}, the authors provide an algorithm for finding non-temporal actual causes for HyperLTL effects. Our approach is incomparable to theirs, not only in the different nature of causes, but in two additional aspects: First, \cite{coenen2022hyperproperty} handles HyperLTL effects while we only handle LTL effects. Second, we provide non-temporal reasons in their trace representations, thus, we can find reasons as infinite sets of events on the trace (see \Cref{fig:msr_output}), while \cite{coenen2022hyperproperty} only provides finite sets of events as causes. 

Since the focus of this work is sufficient reasons, we refer the reader to the \href{https://github.com/SubmissionForPaper1/research-artifact.git}{git repository} for evaluation of \MCEs~synthesis.

\begin{table}[t] % Standard 'table' for one-column papers
    \centering
    
    % --- LEFT TABLE ---
    \begin{minipage}[t]{0.54\textwidth}
        \centering
        % This treats TAB1.tex like a figure and scales it to fit 49% of the page width
        \resizebox{\linewidth}{!}{
            %\begin{table}[t]
%\begin{wraptable}{r}{0.6\textwidth} % CHANGE 1: Use 'wraptable' instead of 'table'. {r} = right side, 0.6\textwidth = width of the wrap.

%\centering 
 %   \small 
  %  \setlength{\tabcolsep}{3pt} % Shrinks the space between columns significantly
   % \renewcommand{\arraystretch}{1.0} % Standard breathing room
    
    \begin{tabular}{@{} l l r c r c @{}}
        \toprule
        & & \multicolumn{2}{c}{\textbf{TSR}} & \multicolumn{2}{c}{\textbf{CATS}} \\
        \cmidrule(lr){3-4} \cmidrule(l){5-6}
        
        \textbf{Instance} & \textbf{Cause} & \textbf{t(s)} & \textbf{Res} & \textbf{t(s)} & \textbf{Res} \\
        \midrule
        
        Arbiter     & $r1$ (odd pos.)                                                               & 0.41 & \checkmark & 20.02 & \checkmark \\
        Arbiter\_spu & $\LTLeventually r1$                                                          & 0.38 & \checkmark & 0.67  & \xmark     \\
        Ex1         & $i1 \land \LTLnext\LTLnext(\neg i2 \mathbin{\mathcal{U}} (i1 \land \neg i2))$ & 0.40 & \checkmark & 1.01  & \checkmark \\
        Ex1\_mod     & $i1 \land \LTLnext\LTLnext(\neg i2 \mathbin{\mathcal{U}} (i1 \land \neg i2))$ & 0.40 & \checkmark & 1.02  & \xmark     \\
        Ex6\_glob    & $\LTLglobally r0$                                                            & 0.40 & \checkmark & 0.74  & \checkmark \\
        Ex8         & $a \lor \LTLnext a$                                                          & 0.42 & \checkmark & 0.98  & \checkmark \\
        TP\_Left     & $\neg a$                                                                     & 0.38 & \checkmark & 0.52  & \checkmark \\
        
        \bottomrule
    \end{tabular}
%\end{table}
%\end{wraptable} % CHANGE 2: Close the 'wraptable' instead of 'table'
        }
         \captionof{table}{ \textbf{Verification of \TSRs.} ($\checkmark$) indicates the tool considers it as a cause / reason. We note that CATS consider \emph{actual causality} while we (TSR) consider \SRs, therefore, it might be that one tool considers a formula as a correct explanation while the other does not. }
    \label{tab:cats_results1}
    \end{minipage}
    \hfill % This pushes the tables apart to use the full width of your single column
    % --- RIGHT TABLE ---
    \begin{minipage}[t]{0.44\textwidth}
        \centering
        % Scaled to fit the other 49% of the page width
        \resizebox{\linewidth}{!}{
            %\begin{table}[t]
  %  \centering
   % \footnotesize
    %\setlength{\tabcolsep}{3.5pt} % Tightens horizontal gutters significantly
    %\renewcommand{\arraystretch}{1.0} % Standard row height for compactness

    \begin{tabular}{@{} l l c c c c @{}}
        \toprule
        & & \multicolumn{2}{c}{\textbf{MTSR }} & \multicolumn{2}{c}{\textbf{CORP} } \\
        \cmidrule(lr){3-4} \cmidrule(l){5-6} 
        
        \textbf{Instance} & \textbf{Effect} & \textbf{$|\mathcal{A}|$} & \textbf{t(ms)} & \textbf{$|\mathcal{B}|$} & \textbf{t(ms)} \\
        \midrule
        
        Arbiter       & $\LTLeventually g_1$ & 3 & 12.3 & 3 & 72.2 \\
        Arbiter\_sim  & $\LTLnext g_1$ & 3 & 12.4 & 3 & 72.8 \\
        Ex1           & $\LTLeventually \text{error}$ & 4 & 12.8 & 4 & 80.9 \\
        Ex6\_glob     & $\LTLglobally \neg g_1$ & 2 & 12.6 & 2 & 75.0 \\
        Ex8\_mod      & $\LTLnext\LTLnext \text{error}$ & 3 & 12.9 & 2 & 74.1 \\
        TP\_Left      & $\LTLeventually e$ & 1 & 12.1 & 1 & 73.2 \\
        
        \bottomrule
    \end{tabular}
   
%\end{table}
        }
         \captionof{table}{\textbf{Synthesizing \MTSRs.} The \emph{Effect} column specifies the property that the generated cause explains. $|\mathcal{A}|$ and $\mathcal{B}$ denote the number of states in the generated NBA for MTSR and CORP, respectively.}
    \label{tab:crop_vs_mtsr}
    \end{minipage}
    \vspace{-3mm}
\end{table}

\begin{comment}
    
\begin{figure}[t]
    \centering
        \includegraphics[width=0.6\linewidth]{MSR_picture_inverted3.png}
    % ----------------------------
    \caption{\textbf{\MSR~synthesis.} The part above CYCLE refers to the prefix of the trace, while the part below CYCLE refers to the $\omega$-repeating suffix. }%This way, we provide representations of infinite sets as non-temporal reasons.}
    \label{fig:msr_output}
\end{figure}
\end{comment}

\section{Related Work}
\emph{Sufficient reasons and contrastive explanations:}
\cite{DBLP:journals/pacmmod/BarceloKRSV25} defined sufficient reasons and contrastive explanations to explain $k$-nearest neighbor classification, and analyzed complexities for different scenarios of the problem of verifying explanations, and finding them. Inspired by their work, and motivated by the need to unify explainability notions between different research areas, we have lifted their definitions to the temporal setting.
Additional works, such as 
\cite{DBLP:conf/aaai/0001I22,DBLP:journals/jolli/DarwicheH23,DBLP:conf/aaai/IgnatievNM19,DBLP:conf/aaai/DarwicheJ22}
consider sufficient reasons and contrastive explanations as formal explanations in the context of AI, which supports our motivation to define them also in the setting of formal verification and temporal specifications.

\emph{Causality and explainability for temporal properties:} 
Many works have studied the explainability problem for temporal properties~\cite{FFMS24,FS23,FFS25,BaierDFJMPZ21,Leitner-FischerL13a}, however non of them studies \SRs~or \TSRs~for reactive systems. Several works studied
 counterfactual reasoning over finite traces or over limited temporal effects \cite{DBLP:journals/sttt/GroceCKS06,DBLP:conf/popl/BallNR03,DBLP:journals/corr/abs-2310-01001}, while we define \CEs~over infinite traces and general LTL effects. 
%Most of these works have considered actual causality~\cite{HalpernP05a}.
\cite{FS23} consider counterfactual theory, and present a new temporal logic to reason about counterfactual worlds. This is orthogonal to our work, that defines (T)\CEs~for reactive systems, and aims at using existing formalisms to reason about them.

The only line of work that have studied in depth both temporal effects and temporal explanations for reactive systems over infinite words are the extensions of Halpern and Pearl's actual causality \cite{HalpernP05a} to the temporal setting. This line of works provided: definitions of non-temporal actual causality \cite{coenen2022hyperproperty}; definitions of temporal actual causality \cite{DBLP:conf/atva/CoenenFFHMS22}; a tool implementation \cite{DBLP:conf/atva/BeutnerFFS23}; synthesis of temporal actual causality \cite{FFMS24}; and complexity analysis \cite{LICS25}. In our work, we approach all of these aspects for other notions of explainability: (M)\SRs~and \CEs, and provide definitions, algorithms, complexity analysis, and a prototype implementation. We note that~\cite{coenen2022hyperproperty} only finds causes as finite sets of events, while we manage to find infinite sets of events, using an ultimately periodic representation. 
%Previous works have considered other notions of temporal causality, such as actual causality of Halpern and Pearl~\cite{cite}, 

%and counterfactual theory~\cite{cite}. 
%In this work, we have lifted known and practically used notions of explanations, provided detailed complexity analysis, that was lacking from almost all previous work on temporal causality, and demonstrated our approach on a set of benchmarks. 

%In contrast to the counterfactual operators suggested by~\cite{FS23}, that defines new counterfactual operators and a new logical framework, here, we are interested in finding simpler solutions, and in their complexity analysis, thus lifting the work of~\cite{DBLP:journals/pacmmod/BarceloKRSV25} to the temporal setting.

\section{Conclusions}

Motivated by the need to formalize different explainability notions and 
unify~these notions across different fields, we study sufficient reasons and contrastive explanations for reactive systems given LTL effects. We provide formal definitions~for the non-temporal and temporal variants of the problem, together with algorithms and complexity analysis for verification and synthesis of the different variants. 

Sufficient reasons provide another view on explainability of reactive systems, and we believe that exploring more notions of explainability, and formally defining them, would allow better understanding of systems and their behavior. 

In the future, we plan to develop our prototype into a fully automated tool that can be used for computing (minimal, temporal) sufficient reasons and contrastive explanations 
automatically. We also plan to study the relations and derivations between the different causality and explainability notions.

%% The file kr.bst is a bibliography style file for BibTeX 0.99c
\bibliographystyle{plainurl}
\bibliography{bib}
%\clearpage

\newpage

\renewcommand{\theHsection}{A\arabic{section}} % Prefixes internal link names with "A"
\appendix

%\section*{Supplementary Material for Paper \#383 Sufficient Reasons and Contrastive Explanations for Reactive Systems}

%For ease of readability, we begin each appendix section in a new page.
%\hadar{check references}

%\hadar{read again, make examples better}

%\hadar{add to future work?}
\section{LTL and HyperLTL}\label{app:HyperLTL}

We provide here the full definition of LTL and its semantics, and then of HyperLTL and its semantics. 

\paragraph{LTL} 
Linear Temporal Logic (LTL~\cite{Pnueli77}) is defined over the set of atomic propositions $\ap$ using the following grammar: 
%\begin{equation*}
\begin{align*}
  \psi := p ~|~ \psi \vee \psi ~|~ \neg\psi ~|~ \LTLnext \psi ~|~ \psi\LTLuntil \psi 
\end{align*}
%\tag{$\star$}\label{eq:hyperLTL}
%\end{equation*}
%
Where $p\in \ap$, $\LTLnext$ is the temporal operator \emph{next}, and $\LTLuntil$ is the temporal operator \emph{until}. We use the usual derived Boolean operators $\wedge, \rightarrow, \leftrightarrow$, and the derived temporal operators $\LTLeventually \psi = \mathit{true}\LTLuntil \psi$  (\emph{eventually} $\psi$) and $\LTLglobally \psi= \neg \LTLeventually \neg \psi$ (\emph{globally} $\psi$). 

We define when an infinite word $\alpha\in (2^\ap)^\omega $
satisfies an LTL formula $\varphi$ inductively, as follows: 
\begin{align*}
       &\alpha \vDash p&\text{ iff } & p\in \alpha[0] \\ 
    &\alpha \vDash \neg \psi &\text{ iff } &\alpha \nvDash \psi \\
    & \alpha\vDash \psi_1 \vee \psi_2  &\text{ iff } &  \alpha \vDash \psi_1 \text{ or } \alpha \vDash \psi_2 \\
       &\alpha \vDash \LTLnext \psi &\text{ iff } & \alpha[1,\infty) \vDash \psi \\ 
      &\alpha \vDash \psi_1 \LTLuntil \psi_2 &\text{ iff } & \exists t_2\in\mathbb{N} \text{ s.t. } \alpha[t_2,\infty) \vDash \psi_2 \\ &&&\text{ and } \forall t_1< t_2: \alpha[t_1,\infty) \vDash \psi_1 
\end{align*}
We say that a system $T$ satisfies an LTL formula $\varphi$, denoted by $T\vDash \varphi$ iff $\forall \alpha \in \TR(T): \alpha\vDash\varphi$.

The language of an LTL formula $\varphi$ is the set of all word that satisfy $\varphi$: 
$\mathcal{L}(\varphi) = \{ \alpha\in (2^\ap)^\omega~|~ \alpha\vDash\varphi \} $. 

It holds that every LTL formula $\varphi$ can be translated to a non-deterministic Büchi automaton $\aut$ such that $\mathcal{L}(\varphi) = \mathcal{L}(\aut)$~\cite{DBLP:conf/lics/VardiW86}. We denote this automaton by $\aut_\varphi$ and will use this notation throughout the paper.

\paragraph{HyperLTL}\label{sec:prelim:hyperQPTL}
HyperLTL~\cite{ClarksonFKMRS14} extends LTL with quantification over traces, 
allowing it
to express relational properties over infinite executions. 

The syntax of HyperLTL 
is defined using the following grammar,
given a finite set $\mathcal{V}$ of trace variables. %We ignore the 
\begin{equation*}
\begin{aligned}
  &\psi := p_\pi ~|~ \psi \vee \psi ~|~ \neg\psi ~|~ \LTLnext \psi ~|~ \psi\LTLuntil \psi \\
  &\varphi := \exists \pi. \varphi ~|~ \forall \pi. \varphi ~|~ \psi  %~|~ \textcolor{blue}{\forall q. \varphi}
\end{aligned}
%\tag{$\star$}\label{eq:hyperLTL}
\end{equation*}
Where $\pi\in \mathcal{V}$ is a trace variable. That is, HyperLTL formulas are LTL formulas where atomic propositions are augmented by traces variables, and, with a prefix of existential and universal trace quantifiers.

The \emph{semantics} of HyperLTL are given with respect to a set of traces $\mathcal{T}$ and a trace assignment $\Pi: \mathcal{V}\rightarrow \mathcal{T}$. Given a trace variable $\pi\in\mathcal{V}$ and a trace $\tau\in\mathcal{T}$, we denote by $\Pi[\pi\mapsto \tau]$ the trace assignment that coincides with $\Pi$ on all variables, except for $\pi$ which is mapped to $\tau$. 
For the assignment $\Pi$, we denote by $\Pi[t, \infty)$ the assignment that maps every variable $\pi \in \mathcal{V}$ to the suffix of $\Pi(\pi)$ starting from position $t$, that is: $\Pi(\pi)[t]\Pi(\pi)[t+1]\Pi(\pi)[t+2]\cdots$. The semantics of HyperLTL is defined inductively as follows. 
\begin{align*}
     &\Pi, \mathcal{T} \vDash p_\pi &\text{ iff } & p\in\Pi(\pi)(0) \\ 
    &\Pi, \mathcal{T} \vDash \neg \psi &\text{ iff } &\Pi,\mathcal{T} \nvDash \psi \\
    & \Pi, \mathcal{T} \vDash \psi_1 \vee \psi_2  &\text{ iff } &  \Pi, \mathcal{T} \vDash \psi_1 \text{ or } \Pi, \mathcal{T} \vDash \psi_2 \\
       &\Pi, \mathcal{T} \vDash \LTLnext \psi &\text{ iff } & \Pi(\pi)[1,\infty),\mathcal{T} \vDash \psi \\ 
      &\Pi, \mathcal{T} \vDash \psi_1 \LTLuntil \psi_2 &\text{ iff } & \exists t_2\in\mathbb{N} \text{ s.t. } \Pi[t_2,\infty),\mathcal{T} \vDash \psi_2 \\ &&&\text{ and } \forall t_1< t_2: \Pi[t_1,\infty),\mathcal{T} \vDash \psi_1 \\ 
      &\Pi, \mathcal{T} \vDash \exists \pi. \varphi &\text{ iff } & \text{ there exists } \tau\in\mathcal{T} \text{ s.t. } \Pi[\pi\mapsto \tau],\mathcal{T}\vDash \varphi \\
       &\Pi, \mathcal{T} \vDash \forall \pi. \varphi &\text{ iff } & \text{ forall } \tau\in\mathcal{T}: \Pi[\pi\mapsto \tau],\mathcal{T}\vDash \varphi 
\end{align*}
We say that $\mathcal{T} $ satisfies $\psi$, denoted $\mathcal{T} \vDash \psi$, if $\Pi_\emptyset, \mathcal{T} \vDash \psi$ where $\Pi_\emptyset: \mathcal{V}\rightarrow \emptyset$ is the empty assignment, meaning that $\psi$ contains no free trace variables that need to be assigned.
We say that a transition system $T$ satisfies a HyperLTL formula~$\psi$ if $\TR(T)$ satisfies $\psi$.

%\newpage
\section{Full Proofs}\label{app:proofs}

\paragraph{\Cref{lemma:RiffTau} (restated)}
Let $R \subseteq \mathbb{N}\times \ip$ and let $\tau_R $
be the trace representation of $R$. 
%, such that: $R = \{ i\in \mathbb{N}~|~ \on\in\tau_R[i]  \}$. 
Then, $f[R, \pi] $ 
holds in $T$
iff $ \psi_{sr}[\tau_R, \pi]$ holds in $T_\aug$. 
%\Cref{eq:reasonR} holds with respect to $\tuple{T, \pi}$ iff \Cref{eq:reasonTau} holds with respect to $\tuple{T_\aug, \pi}$. 

\begin{proof}
    Using $\R_{\tau_R}$ we require that $\tau_R$ is a trace representation of some set $R$.

    We have that $\TR(T_\aug) \cap (2^\ap)^\omega = \TR(T)$. In particular, for $\pi' \in \TR(T_\aug)$, we have  $\neg \R_{\pi'} \leftrightarrow \pi'\in\TR(T)$, thus, in $ \psi_{sr}[\tau_R, \pi]$ we only consider traces $\pi'$ of the system $T$. 
    The condition 
    $  \LTLglobally   \bigwedge_{a\in \ip}
    \big(   
    a_{\tau_R} \rightarrow (\pi =_{\{a\}} \pi')
    \big)$
   % $ \LTLglobally \big(   \on_{\tau_R} \rightarrow (\pi =_\ip \pi')    \big)$ 
    is equivalent to 
    $\forall (i,a)\in R. \pi'[i] =_{\{a\}} \pi[i]$,
    %$\forall i. i\in R \rightarrow \pi'[i]=_\ip \pi[i]$, 
    which is exactly the condition in the definition of $f[R,\pi]$. \qed
\end{proof}

\paragraph{\Cref{thm:pspacetemporal} (restated)}
    Given a trace property $P_R$, as a Büchi automaton
    or an LTL formula,
    verifying if $P_R$ is a temporal sufficient reason for an LTL effect $\varphi$ in $\tuple{T, \pi}$ is PSPACE-complete. 

\begin{proof} 
    Membership in PSPACE: 
First, assume that $P_R$ is given as an NBA over $2^\ip$.
    Let $\mathcal{A}_{P_R}$ be the extension of $P_R$ to the full set $2^\ap$ (c.f. \Cref{def:extension}).
    Then,   \Cref{eq:buchiP} reduces to model-checking whether $\mathcal{A}_{P_R} \cap T \vDash\mathcal{A}_\varphi$.

    If $P_R$ is an LTL property $\varphi_R$,  then, to verify that $\varphi_R$ is a temporal sufficient reason for $\varphi$ in $\tuple{T,\pi}$ we need to model-check \Cref{eq:TSR} against $T$, which is purly an LTL model-checking problem.

    PSPACE hardness: we can reduce the LTL model-checking problem to the verification of temporal sufficient reasons, with almost the same proof as the proof of \Cref{thm:SRpspacecomplete}, but with the change that we are given the trace property $P_R = \mathit{true}$ rather then the empty set of events (an automaton that accepts all words has one state, and an LTL formula is of constant size):
       $\mathit{true}$ is a temporal sufficient reason for some LTL property $\varphi$ in $\tuple{T, \pi}$ iff all traces of $T$ satisfy $\varphi$. We now provide the full proof following the idea of \Cref{thm:SRpspacecomplete}.

        We show a reduction from LTL model-checking: Given a transition system $T$, and an LTL formula $\varphi$, we construct an instance: $T'$, $\pi$, $\varphi'$ and $P_R$ in polynomial time, such that $T\models \varphi$ iff $P_R$ is a sufficient reason for $\varphi'$ in $\tuple{T', \pi}$.
    \begin{itemize}[nosep]
        \item We construct $T'$ from $T$ by adding a fresh state $q_\fresh \notin Q$, and the transitions 
        $\delta(q_0, \{ \fresh\})=( q_\fresh, \emptyset)$, $\delta(q_\fresh, \{ \fresh\})=( q_\fresh, \emptyset)$
        %$(q_0, \{ \fresh\}, q_\fresh), (q_\fresh, \{ \fresh\}, q_\fresh)$ 
        for a fresh input proposition $\fresh \notin \ap$. 
     We denote $\pi_\fresh = \{\fresh\}^\omega$, and
    we have that that $\TR(T') = \TR(T) \cup \{ \pi_\fresh \}$.
    \item  We define $\varphi' := \varphi \vee \LTLglobally \fresh$. Then, we have the following:
     \begin{itemize}[nosep]
         \item $\pi_\fresh \vDash \varphi'$. 
         \item $\forall \pi' \in \TR(T)$: $\pi'\nvDash \LTLglobally \fresh$. Therefore: 
         \item $\forall \pi' \in \TR(T)$: $\pi'\vDash \varphi $ iff $\pi'\vDash \varphi' $.
     \end{itemize}
     \item We choose $\pi = \pi_\fresh$.
     \item We choose $P_R = \true$.
    \end{itemize}
    The construction of $T'$ is linear in $T$ (both in terms of time and space), and so is the construction of $\varphi'$. $\pi_\fresh$ and $P_R$ are of constant size, so all together we have a polynomial construction. 

    We are now left to show that 
     $T\models \varphi$ iff $\true$ is a temporal sufficient reason for $\varphi'$ in $\tuple{T', \pi_\fresh}$. We have that 
     $\true$ is a temporal sufficient reason for $\varphi'$ in $\tuple{T', \pi_\fresh}$ iff $\varphi'$ holds on all traces in $\TR(T')$. Since $\varphi'$ holds on $\pi_\fresh$ regardless of $T$ and $\varphi$, and from the construction of $T'$, we have that $\varphi'$ holds on all traces in $\TR(T')$ iff $\varphi$ holds on all traces in $\TR(T)$, as needed. \qed
\end{proof}

\paragraph{\Cref{thm: msrs vs ces} (restated)}
   Let $\pi$ be an ultimately periodic trace, let $\varphi$ be an LTL effect and let $T$ be a transition system. Let $\mathfrak{R} = \{ R_i ~|~ i\in\mathfrak{I} \}$ be the set of all \MSRs~for $\phiTpi$. Note that $\mathfrak{R}$ might be uncountable, or empty. Then, for any \MCE~$\tau$ we have  
    $\forall i\in \mathfrak{I}: \Delta(\tau,\pi) \cap R_i \neq \emptyset$. 

  %  \item Let $S \subseteq \ip\times\nat$ such that   $\forall i\in \mathfrak{I}: S \cap R_i \neq \emptyset$. Then, there is a trace $\tau$ such that $\Delta(\tau,\pi) = S$ and such that $\tau$ is an \MCE~for $\phiTpi$.

\begin{proof}
       Let $\tau$ be an \MCE, and assume by a way of contradiction that there is $i\in\mathfrak{I}$ such that $R_i \cap \Delta(\tau,\pi) = \emptyset$. This means that $\pi$ agrees with $\tau$ on all events in $R_i$, and since $R_i$ is an \MSR, we have that $\tau \vDash \varphi$, contradicting the fact that $\tau$ is a \CE. \qed
\end{proof}

\newpage
\section{Additional Tables for Experimental Evaluation} \label{app:tables}
\begin{figure}[H]
    \centering
            \includegraphics[width=0.97\linewidth]{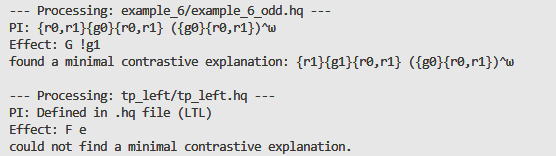}

    \caption{MCE: \textbf{Synthesizing a Minimal Contrastive Explanation.} These results align with the definitions.}
    \label{fig:mce_outputs}
\end{figure}
\vspace{-3mm}

\begin{table}[H]
    \centering
    \small % Fits better in columns
    \setlength{\tabcolsep}{2pt} % Tight columns
    \renewcommand{\arraystretch}{1.1} % slightly more breathing room
    
    % Use tabularx to let the 'Cause' column fill remaining space
    \begin{tabularx}{\columnwidth}{@{} l >{\raggedright\arraybackslash}X r c r c @{}}
        \toprule
        & & \multicolumn{2}{c}{\textbf{TSR}} & \multicolumn{2}{c}{\textbf{CATS}} \\
        \cmidrule(lr){3-4} \cmidrule(l){5-6} % Nice partial lines
        
        \textbf{Instance} & \textbf{Cause} & \textbf{t (s)} & \textbf{Res} & \textbf{t (s)} & \textbf{Res} \\
        \midrule
        
        Arbiter          & $r1$ (odd pos.)                                              & 0.41 & \checkmark     & 20.02 & \checkmark \\
        Arbiter\_sim     & $\LTLnext r1$                                                       & 0.38 & \checkmark & 0.95  & \checkmark \\
        Arbiter\_spu     & $\LTLeventually r1$                                          & 0.38 & \checkmark     & 0.67  & \xmark     \\
        Ex1              & $i1 \land \LTLnext\LTLnext(\neg i2 \mathbin{\mathcal{U}} (i1 \land \neg i2))$ & 0.40 & \checkmark & 1.01  & \checkmark \\
        Ex1\_mod         & $i1 \land \LTLnext\LTLnext(\neg i2 \mathbin{\mathcal{U}} (i1 \land \neg i2))$ & 0.40 & \checkmark & 1.02  & \xmark     \\
        Ex6\_glob        & $\LTLglobally r0$                                            & 0.40 & \checkmark & 0.74  & \checkmark \\
        Ex6\_odd         & $r0$ (even pos.)                                             & 0.41 & \checkmark & 2.32  & \checkmark \\
        Ex8              & $a \lor \LTLnext a$                                                 & 0.42 & \checkmark & 0.98  & \checkmark \\
        Ex8\_mod         & $a \lor \LTLnext a$                                                 & 0.41 & \checkmark & 0.99  & \checkmark \\
        TP\_Left         & $\neg a$                                                     & 0.38 & \checkmark & 0.52  & \checkmark \\
        TP\_Right        & $\neg a$                                                     & 0.98 & \checkmark & 1.86  & \checkmark \\
        
        \bottomrule
    \end{tabularx}
    \caption{
    Full version of \Cref{{tab:cats_results1}}: 
     \textbf{Verification of Temporal Sufficient Reasons (TSR).} The checkmark ($\checkmark$) indicates the tool considers it as a cause / reason. We note that CATS consider \emph{actual causality} while we (TSR)  consider \emph{sufficient reasons}, therefore, it might be that one tool considers a formula as a correct explanation while the other does not. }%\footnotetext{\cite{DBLP:conf/atva/BeutnerFFS23}}
\end{table}
\vspace{-3mm}

\begin{table}[h]
    \centering
    \footnotesize
    \setlength{\tabcolsep}{6pt} % Balanced spacing
    \renewcommand{\arraystretch}{1.4}

    \begin{tabularx}{\columnwidth}{@{} l l c c c c @{}}
        \toprule
        & & \multicolumn{2}{c}{\textbf{MTSR }} & \multicolumn{2}{c}{\textbf{CORP} } \\
        \cmidrule(lr){3-4} \cmidrule(l){5-6} 
        
        \textbf{Instance} & \textbf{Effect} & \textbf{$|A|$} & \textbf{Time (ms)} & \textbf{$|B|$} & \textbf{Time (ms)} \\
        \midrule
        
        Arbiter       & $\LTLeventually g_1$ & 3 & 12.3 & 3 & 72.2 \\
        Arbiter\_sim  & $\LTLnext g_1$ & 3 & 12.4 & 3 & 72.8 \\
        Arbiter\_spu  & $\LTLeventually g_1$ & 1 & 12.0 & 1 & 72.9 \\
        Ex1           & $\LTLeventually \text{error}$ & 4 & 12.8 & 4 & 80.9 \\
        Ex1\_mod      & $\LTLeventually o$ & 4 & 12.9 & 4 & 75.1 \\
        Ex6\_glob     & $\LTLglobally \neg g_1$ & 2 & 12.6 & 2 & 75.0 \\
        Ex6\_odd      & $\LTLglobally \neg g_1$ & 2 & 12.6 & 2 & 75.6 \\
        Ex8           & $\LTLnext\LTLnext \text{error}$ & 3 & 12.7 & 3 & 74.1 \\
        Ex8\_mod      & $\LTLnext\LTLnext \text{error}$ & 3 & 12.9 & 2 & 74.1 \\
        TP\_Left      & $\LTLeventually e$ & 1 & 12.1 & 1 & 73.2 \\
        TP\_Right     & $\LTLeventually e$ & 2 & 12.1 & 2 & 74.5 \\
        
        \bottomrule
    \end{tabularx}
    \caption{
    Full version of  \Cref{tab:crop_vs_mtsr}:
    \textbf{Synthesizing a Minimal Temporal Sufficient Reason.} The \emph{Effect} column specifies the property that the generated cause explains. $|A|$ and $|B|$ denote the number of states in the generated witness automaton for MTSR and CORP, respectively.}
\end{table}

\begin{table}[H]
    \centering
    \footnotesize
    \renewcommand{\arraystretch}{1.4} % More vertical space for formulas
    \setlength{\tabcolsep}{3pt}

    \begin{tabularx}{\columnwidth}{@{} l >{\raggedright\arraybackslash}X >{\raggedright\arraybackslash}X @{}}
        \toprule
        \textbf{Instance} & \textbf{Cause $L(A)$ (MTSR)} & \textbf{Cause $L(B)$ (CORP)} \\
        \midrule
        
        Arbiter       & $\neg r_1 \land \neg r_0 \land \LTLnext(\LTLeventually r_1)$ & NBA \\
        Arbiter\_sim  & $\neg r_1 \land \neg r_0 \land \LTLnext r_1$ & $\LTLnext r_1$ \\
        Arbiter\_spu  & TRUE & TRUE \\
        Ex1           & $i_1 \land \LTLnext\LTLnext(\neg i_2 \mathcal{U} (i_1 \land \neg i_2))$ & $i_1 \land \LTLnext\LTLnext(\neg i_2 \mathcal{U} (i_1 \land \neg i_2))$ \\
        Ex1\_mod      & $i_1 \land \LTLnext\LTLnext(\neg i_2 \mathcal{U} (i_1 \land \neg i_2))$ & $i_1 \land \LTLnext\LTLnext(i_1 \land \neg i_2)$ \\
        Ex6\_glob     & NBA & $r0$ (even pos.) \\
        Ex6\_odd      & NBA & $r0$ (even pos.) \\
        Ex8           & $a \lor \LTLnext a$ & $a \lor \LTLnext a$ \\
        Ex8\_mod      & $a \lor \LTLnext a$ & $a$ \\
        TP\_Left      & $\LTLglobally \neg a$ & TRUE \\
        TP\_Right     & $\neg a$ & $\neg a$ \\
        
        \bottomrule
    \end{tabularx}
    \caption{\textbf{Synthesizing a Minimal Temporal Sufficient Reason.}
 Comparison of the logical formulas derived from the automata of \Cref{tab:crop_vs_mtsr}. Note that in some instances, the cause can only be represented as an NBA and does not have an equivalent LTL formula.
    Both our MTSR implementation and the CORP tool output explicit NBAs that characterize the language of the minimal temporal sufficient reason (MTSR)/actual temporal cause (CORP). 
The LTL formulas were derived manually, to provide a 
%To facilitate 
a direct semantic comparison (as done
in n (Finkbeiner et al. 2024)
to evaluate CORP).  }
\end{table}

\begin{table}[t]
    \centering
    \footnotesize
    \setlength{\tabcolsep}{6pt}
    \renewcommand{\arraystretch}{1.6}
    \setcounter{table}{5}% render this as Table 6
    \ifdefined\mcetblbox\else\newsavebox{\mcetblbox}\fi
    \sbox\mcetblbox{%
    \begin{tabular}{@{} l l l l c c @{}}
        \toprule
        \textbf{Instance} & \textbf{Effect} & \textbf{\Large $\pi$} & \textbf{MCE} & \textbf{$|\mathrm{MCE}|$} & \textbf{t (s)} \\
        \midrule
        Ex6\_glob & $\LTLglobally \neg g_{1}$ & $(\{r_{0},r_{1}\}\{g_{0}\})^{\omega}$ & $\{r_{1}\}\{g_{1}\}\,(\{r_{0},r_{1}\}\{g_{0}\})^{\omega}$ & 4 & 1.02 \\
        Ex6\_odd & $\LTLglobally \neg g_{1}$ & $(\{r_{0},r_{1}\}\{g_{0}\})^{\omega}$ & $\{r_{1}\}\{g_{1}\}\,(\{r_{0},r_{1}\}\{g_{0}\})^{\omega}$ & 4 & 0.95 \\
        TP\_Left & $\LTLeventually e$ & $(\{e\})^{\omega}$ & \textemdash & - & 0.81 \\
        TP\_Right & $\LTLeventually e$ & $(\{e\})^{\omega}$ & $\{a\}\,(\{\})^{\omega}$ & 2 & 0.82 \\
        \bottomrule
    \end{tabular}%
    }%
    \begin{minipage}{\wd\mcetblbox}
        \centering
        \usebox\mcetblbox
        \caption{\textbf{Synthesizing a Minimal Contrastive Explanation (MCE).} $\pi$ is the original trace where the Effect holds; \emph{MCE} is the synthesized minimal counterfactual trace $\tau$ (\textemdash{} when no counterfactual exists). }
    \end{minipage}
\end{table}
\begin{table}[t]
    \centering
    \footnotesize
    \setlength{\tabcolsep}{6pt}
    \renewcommand{\arraystretch}{1.6}
    \setcounter{table}{6}% render this as Table 7
    \ifdefined\msrtblbox\else\newsavebox{\msrtblbox}\fi
    \sbox\msrtblbox{%
    \begin{tabular}{@{} l l l c c @{}}
        \toprule
        \textbf{Instance} & \textbf{Effect} & \textbf{MSR} & \textbf{$|R|$} & \textbf{t (s)} \\
        \midrule
        Arbiter & $\LTLeventually g_{1}$ & $\{(r_{1},3)\}$ & 1 & 6.45 \\
        Arbiter\_sim & $\LTLnext g_{1}$ & $\{(r_{1},1)\}$ & 1 & 7.01 \\
        Arbiter\_spu & $\LTLeventually g_{1}$ & $\emptyset$ & 0 & 4.83 \\
        Ex1 & $\LTLeventually o$ & $\{(i_{1},0), (i_{2},2), (i_{2},3), (i_{1},4), (i_{2},4)\}$ & 5 & 4.88 \\
        Ex1\_mod & $\LTLeventually o$ & $\{(i_{1},0), (i_{1},2), (i_{2},2)\}$ & 3 & 5.17 \\
        Ex8 & $\LTLnext \LTLnext \text{error}$ & $\{(a,1)\}$ & 1 & 1.11 \\
        Ex8\_mod & $\LTLnext \LTLnext \text{error}$ & $\{(a,0)\}$ & 1 & 1.14 \\
        Ex6\_glob & $\LTLglobally \neg g_{1}$ & $\{(r_0,2i) : i \in \mathbb{N}\}$ & $\infty$ & 5.61 \\
        Ex6\_odd & $\LTLglobally \neg g_{1}$ & $\{(r_0,2i) : i \in \mathbb{N}\}$ & $\infty$ & 5.62 \\
        TP\_Left & $\LTLeventually e$ & $\emptyset$ & 0 & 1.03 \\
        TP\_Right & $\LTLeventually e$ & $\{(a,0)\}$ & 1 & 1.07 \\
        \bottomrule
    \end{tabular}%
    }%
    \begin{minipage}{\wd\msrtblbox}
        \centering
        \usebox\msrtblbox
        \caption{\textbf{Synthesizing a Minimal Sufficient Reason (MSR).} The \emph{Effect} column is the property being explained. \emph{MSR} is the synthesized reason $R \subseteq \mathrm{IP}\times\mathbb{N}$ (input propositions $\times$ time steps); $|R|$ is its size, $\infty$ when the reason repeats on the loop.}
    \end{minipage}
\end{table}

\end{document}